\documentclass[10pt,a4paper]{amsart}

\usepackage{latexsym,amssymb,amsmath,amsthm,amsfonts,enumerate,verbatim,xspace,
exscale}
\usepackage{graphicx}
\usepackage{color,amsbsy,textcomp}
\usepackage{enumerate}
\usepackage{float} 

\usepackage{orcidlink} 

\usepackage{hyperref}
\hypersetup{
    colorlinks=true,
    linkcolor=blue,
    filecolor=magenta,      
    urlcolor=cyan,
}

\input xy 
\xyoption{all} 
\CompileMatrices
\UseComputerModernTips

\theoremstyle{plain}
\newtheorem{theorem}{Theorem}[section]
\newtheorem{lemma}[theorem]{Lemma}

\theoremstyle{definition}
\newtheorem{definition}[theorem]{Definition}
\newtheorem{remark}[theorem]{Remark}

\theoremstyle{definition}
\newtheorem{ass}[theorem]{Assumption}
\newtheorem{mr}[theorem]{Rule}

\def\R{\mathbb{R}}

\def\N{\mathbb{N}}

\newcommand{\di}{\mbox{$\,\textup{d}$}}

\newcommand{\up}{\upshape}

\newcommand{\longto}{\longrightarrow}

\def\vv<#1>{\langle#1\rangle}

\newcommand{\id}{\mbox{$\text{\up{id}}\,$}}

\newcommand{\om}{\omega}
\newcommand{\Om}{\Omega}

\newcommand{\eps}{\varepsilon}
\newcommand{\lam}{\lambda}

\newcommand{\todo}[1]{\phantom{u}\vspace{5 mm}\par \noindent
\marginpar{\textsc{ToDo}} \framebox{\begin{minipage}[c]{0.95
\textwidth}\raggedright \tt #1 \end{minipage}}\vspace{5 mm}\par}

\newcommand{\revise}[1]{#1}

\usepackage[normalem]{ulem}
\usepackage{xcolor}

\newcommand\deleteF{\bgroup\markoverwith{\textcolor{blue}{\rule[0.8ex]{2pt}{0.9pt}}}\ULon}
\newcommand\deleteS{\bgroup\markoverwith{\textcolor{green}{\rule[0.8ex]{2pt}{0.9pt}}}\ULon}

\usepackage[foot]{amsaddr}

\title[MF-LMM and long term guarantees]%
{ 
Mean-field Libor market model\\ 
and\\ 
valuation of long term guarantees 
}

\author[Florian Gach, Simon Hochgerner, Eva Kienbacher, Gabriel Schachinger]{Florian Gach\,\orcidlink{0000-0003-2835-8226}${}^1$, Simon Hochgerner\,\orcidlink{0000-0002-3978-3706}${}^{*,1}$, Eva Kienbacher${}^2$, Gabriel Schachinger${}^1$ 
}

\address{%
${}^1$ Austrian Financial Market Authority (FMA),  Otto-Wagner Platz 5, A-1090 Vienna;
} 

\address{${}^2$ %
Oberösterreichische Versicherung AG, Gruberstraße 32, A-4020 Linz; 
} 

\address{${}^*$ Corresponding author;}

\email{Florian.Gach@gmail.com}
\email{Simon.Hochgerner@fma.gv.at}
\email{GabrielSchachinger@gmail.com}
\email{e.kienbacher@ooev.at}

\thanks{\emph{Disclaimer.} 
The opinions expressed in this article are those of the authors and do not necessarily reflect the official position of the Austrian Financial Market Authority. } 
\date{\revise{March 5, 2025}}

\begin{document}

\begin{abstract}
Existence and uniqueness of solutions to the multi-dimensional mean-field Libor market model (introduced by \cite{MFLMM}) is shown. This is used as the basis for a numerical asset-liability management (ALM) model capable of calculating future discretionary benefits in accordance with Solvency~II regulation. This ALM model is complimented with aggregated life insurance data to perform a realistic numerical study.   
This yields numerical evidence for  heuristic assumptions which allow to derive estimators of lower and upper bounds for future discretionary benefits. These estimators are applied to publicly available life insurance data. 
\end{abstract}

\maketitle

\tableofcontents

\section*{Introduction}
This paper is about the market consistent valuation of life insurance with profit participation. These types of products are characterized by a minimum guarantee rate and a profit declaration mechanism that lets the policyholder participate in the company's net profits. The task of assigning a market consistent value to profit participating life insurance liabilities is non-trivial for many reasons, major challenges in designing an appropriate (numerical) model are: 
\begin{itemize}
\item 
Profit is defined in terms of locally generally accepted accounting principles (local GAAP) which implies that assets have to be modelled not only with their market values but also with their book values. 
\item 
Profit declaration is not completely regulated by legislature, thus management rules have to be taken into account. 
\item 
The above points mean that any such model will consist of many subroutines, and there is no realistic hope for a closed formula solution for market consistent valuation. Hence an economic scenario generator (consisting, in particular, of a choice of an interest rate model) has to be fixed, and the valuation procedure depends on Monte Carlo methods. 
\item 
Not least, life insurance is a long term business, and an an appropriate model may have a projection horizon of $60$ years, or even more. 
\end{itemize} 

In the present work we consider the market consistent value to be defined in accordance with Solvency~II which is the relevant regulation for insurers operating in the  European Economic Area. 
This means that the market consistent value is equal to the \emph{technical provisions} as defined in \cite{L1,L2}. The technical provisions are a sum of  a \emph{best estimate}, $BE$, and a risk margin, $RM$. The best estimate, in turn, can be represented as a sum of \emph{guaranteed benefits} and \emph{future discretionary benefits}, $BE = GB+FDB$. All of these quantities are reporting items (\cite{L3templ}), meaning they are part of the Solvency~II balance sheet. Moreover, these quantities are not only reported to national supervisory agencies but it is also mandatory to make these publicly available in the so-called \emph{solvency and financial conditions report (SFCR)}. A prototypical example consisting of partial balance sheet data is provided in Table~\ref{table:A2022}. This table lists liabilities, of which $GB$ and $FDB$, associated to profit participating business, are the two largest items, assets, and own funds, which are the difference between assets and liabilities. It is noteworthy that the future discretionary benefits are larger than the own funds, whence any error in $FDB$ calculation translates to an even bigger error in the own funds, and further in the solvency ratio $OF/SCR$. 

\begin{table}[H]
\centering
\begin{tabular}{lr}
\hline
Best Estimate Life w PP & 203.1 bEUR  
\\
$\phantom{XXX}GB$                       & 163.6 bEUR          
\\
$\phantom{XXX}FDB$              & 39.5   bEUR      
\\
\hline
Liabilities  & 234.5 bEUR   
\\
Assets                     & 269.0 bEUR 
\\
\hline
Own funds ($OF$)                    &  34.1\, bEUR   \\
Solvency capital requirement ($SCR$)                     &  8.2\, bEUR   \\
Solvency ratio ($OF/SCR$)                    & $416\,\%$                                             \\
\hline
\end{tabular}
\vspace{2ex} 
\caption{Solvency~II balance sheet YE 2022: Allianz Lebensversicherung AG (Germany) Solvency and Financial Condition Report (SFCR) 2022; values in billion Euro; public data (\cite{SFCR})
.
}
\label{table:A2022}
\end{table}

The goals of this paper are to, firstly, add to the general understanding of $FDB$-calculation, and secondly, provide tools for fast and reliable $FDB$-validation. In this regard the contributions are threefold: 

\begin{enumerate}
\item 
Advocating the mean-field Libor market model (MF-LMM) that was introduced in \cite{MFLMM}, we prove an existence and uniqueness result which allows us to use this model as an interest rate scenario generator. As shown in a numerical study in \cite{MFLMM} the MF-LMM can be used to reduce the blow-up probability troubling the Libor market model, and is therefore well-suited for numerical valuation of long term guarantees.  However, in \cite{MFLMM} the authors were able to show existence and uniqueness of solutions only in a one-dimensional special case. In Section~\ref{sec:mflmm} we extend existence and uniqueness of MF-LMM to the multi-dimensional setting under rather general assumptions on the coefficients (Theorem~\ref{thm:mflmm}). This result relies on a straightforward application of the $L$-derivative introduced by P.-L.\ Lions (\cite{Car12}). 
\item 
Section~\ref{sec:numALM} contains a description of a numerical asset-liability management (ALM) model that is capable of calculating a realistic $FDB$ value in the Solvency~II sense, and uses the MF-LMM as an interest rate scenario generator. The focus in this section is on a detailed description of modelled management rules. Some of the notation and concepts that are used in the numerical Section~\ref{sec:numALM} are introduced in the previous Section~\ref{sec:li} in a more abstract context
\item 
In Section~\ref{sec:assump} we use the ALM model from Section~\ref{sec:numALM} to provide numerical evidence for phenomenological assumptions concerning various quantities that are involved in the Monte Carlo based $FDB$-calculation. These assumptions follow solely from heuristic ideas, and have no mathematical proofs, however the numerical evidence allows us to conclude that the assumptions are quite reliable -- at least for the (rather wide) range of parameters that we have tested. 
These assumptions are used in Section~\ref{sec:est} to derive algebraic formulas to estimate lower and upper bounds, $\widehat{LB}$ and $\widehat{UB}$, for $FDB$. These bounds yield an estimate $\widehat{FDB} = (\widehat{LB}+\widehat{UB})/2$ for $FDB$, which is useful when the potential error $\eps = (\widehat{UB}-\widehat{LB})/2$ is sufficiently small (e.g.\ compared to the market value of assets). This estimation is compared to the Monte-Carlo value of $FDB$ in Table~\ref{tab:num_res} using the numerical ALM and anonymized insurance data. Finally, in Section~\ref{sec:All_est} we compare the estimation result $\widehat{FDB}$ to publicly available $FDB$-figures, and show how to calculate  $\widehat{FDB}$ using only publicly available data. 
\end{enumerate}

The estimation formulae for $\widehat{LB}$ and $\widehat{UB}$ are a continuation of the work begun in \cite{HG19,GH22} by two of the authors. Compared to \cite{GH22} we have simplified the set of assumptions, improved and simplified the calculation of the lower bound, and, crucially, tested the assumptions on a numerical model (i.e.\ the ALM model of Section~\ref{sec:numALM}). 

Further background on  market consistent valuation of life insurance products can be found in  \cite{SS04, OBrien09, Delong,  W16, Gerstner08, Gerstner09, Bacin21, Engsner_etal17, FN21}). These works all address the interplay between asset and liability modelling but present strongly simplified versions of asset portfolios, in particular book values of assets are not modelled. Besides \cite{HG19,GH22}, the case for a realistic representation of (asset) accounting principles is made in  \cite{D_etal17,VELP17,Albrecher_etal18, Dorobantu_etal20}. The present paper follows this logic, and shows that accounting principles and management rules can be formulated in a manner amenable to mathematical analysis such that general formulae for lower and upper bounds of future discretionary benefits may be derived. As mentioned, these formulae depend on a number of assumptions. The overarching principle behind these assumptions is that the management strives to maximize the shareholder value while remaining a competitive player in the insurance market. This is reminiscent of an optimal control problem, and is formulated as such in Section~\ref{sec:con}. The logical next step would be to show that the assumptions of Section~\ref{sec:assump} define (or: are in an appropriate sense close to) the optimal strategy. However, this conclusion is outside the scope of this work, and is left open as a problem for the future.

\section{Mean-field Libor market model}\label{sec:mflmm}
To simulate market values of bonds, or other interest rate dependent assets, and discount factors, a stochastic interest rate model is needed. To this end, a popular choice in the insurance industry is the Libor market model (e.g., \cite{BM06,Fil}). A well-known downside of  this model is, however,  that  it tends to suffer from blow-up (\cite{G11}). This can have undesirable consequences in conjunction with a numerical ALM model (Section~\ref{sec:numALM}): when forward rates are too high over a prolonged period of time some quantities (e.g., the balance sheet item of declared bonuses, $DB$) may become so large that precision differences occur, and this can in turn trigger undesired and unrealistic management actions or even cause leakage: for example, when a rounding error between assets and liabilities is blown up such that the book value of assets does no longer suffice to cover the statutory value of liabilities management is required to inject additional capital; such a management action is then undesired because it causes leakage, and not realistic because it is triggered by a numerical artefact which does not have a counterpart in reality.  

With the aim of reducing the probability for exploding rates, while at the same time retaining the main features of the Libor market model, \cite{MFLMM} introduced a mean-field Libor market model (MF-LMM). In this model the growth of the second moment is controlled via a mean-field interaction.  

\subsection{MF-LMM: model formulation}
Let us fix a tenor structure $0 = t_0 < t_1 = \delta < \ldots < t_N = N\delta$ where $N\in\N$ and $\delta>0$ are fixed. In the ALM model we consider yearly time steps such that $\delta=1$ and $N$ corresponds to the projection horizon $T$ of Section~\ref{sec:li}. 
For $i=0,\dots,N-1$ the $i$-th forward (Libor) rate valid on $[t,t_{i+1}]$ as seen at time $t \le t_i$ is
\begin{equation*}
F_t^{i}
:= 
\frac{1}{\delta} 
\frac{P(t,t_{i}) - P(t,t_{i+1})}{P(t,t_{i+1})}\,
\end{equation*} 
where $P(t,t_i)$ is the value at $t$ of $1$ currency unit paid at $t_i$.

In the following, we assume that each Brownian motion is defined on a complete filtered probability space which satisfies the usual assumptions (i.e.\ completeness of filtration and right-continuity). 

The classical Libor market model (e.g., \cite{BM06,Fil}) is now specified as follows. For each index $i=0,\dots,N-1$, there should be an $\R^d$-valued diffusion coefficient $\sigma_i = \sigma_i(t)$, a  forward measure $\mathcal{Q}^i$ and a corresponding $d$-dimensional Brownian motion $W^i$ ($d\ge1$), such that the dynamics of $F^i$ under $\mathcal{Q}^i$ are given by 
\begin{align}
	\label{e:LMM}
	dF_t^{i} 
	= F_t^{i}  \sigma_i^{\top} \,dW^{i}_t
	, \qquad t\le t_{i}
\end{align}
where the initial condition is given by the prevailing yield curve at $t=0$. 

Let  $\mathcal{P}(\R^k)$ be the space of probability measures on $\R^k$ and  $\mathcal{P}_2(\R^k) 
= \{\mu\in\mathcal{P}(\R^k) : \int_{\mathbb{R}^k}\vv<x,x>\,\mu(dx)<\infty\}$
the subspace of probability measures with finite second moment.

The idea of the mean-field extension  is to allow $\sigma_i$ to depend not only on time but also on the law of the process. 
Namely, we assume that $\sigma_i: [0,t_{i}]\times\mathcal{P}_2(\R)\to\R^d$:
\begin{align}
    \label{e:LMM-mf}
    \sigma_i 
    &= \sigma_i\Big(t,\mu_{t}^{i}\Big)
    = \lambda_i\Big(t, \mu_t^i(h_i^1), \dots,\mu_t^i(h_i^l)\Big) 
\end{align}
where:
\begin{itemize}
    \item 
    $\mu_{t}^{i} = \operatorname{Law}(F_t^{i}) = (F_t^i)_{\sharp}\mathcal{Q}^i 
    = \mathcal{Q}^i \circ (F_t^i)^{-1}$ 
    is the law of $F_t^i$ under $\mathcal{Q}^i$,
    i.e., the push-forward of $\mathcal{Q}^i$ with respect to $F_t^i$;
    \item 
    $\lambda_i: \R_+\times\R^l\to\R^d$ is a vector valued function with components $\lam_i^j$, $j=1,\dots,d$; 
    \item 
    $h_i: \R \to \R^{l}$ is a vector valued function with components $h_i^k$, $k=1,\dots,l$;
    \item 
    $\mu_t^i(h_i^k) 
    = \int_{\mathbb{R}}h_i^k(x)\,\mu_t^{i}(dx)$
\end{itemize}

Conditions on $\lam_i$ and $h_i$ for the existence and uniqueness of solutions to the mean-field system~\eqref{e:LMM}-\eqref{e:LMM-mf} are given in Theorem~\ref{thm:mflmm} below. 

The dimension $l$ may be chosen arbitrarily following the modelling needs. We will be mostly concerned with the case $l=2$.  
Indeed, a particularly interesting special case arises when $h_1(x) = x$, $h_2(x) = x^2$, independently of $i$, and $\lam_i = \lam_i(t, \mu_t^i(h_1), \mu_t^i(h_2))$.  Notice that $\mu_t^i(h_2) - \mu_t^i(h_1)^2 = \operatorname{Var}_{\mathcal{Q}^i}[F_t^{i}] = \operatorname{Var}_{\mu_t^i}$ is the variance of $F_t^i$ with respect to the forward measure $\mathcal{Q}^i$. 
By slight abuse of notation, we will write 
\[ 
 \lam_i = \lam_i\Big(t, \mu_t^i(h_2) - \mu_t^i(h_1)^2\Big)
\] 
in this case such that $\lam_i: \R_+\times\R\to\R^d$.  That is,   we consider the model 
\begin{align}
    \label{e:LMM-mf-var}
    \sigma_i 
    &= \sigma_i\Big(t,\mu_{t}^{i}\Big)
    = \lambda_i\Big(t,\operatorname{Var}_{\mathcal{Q}^i}[F_t^{i}]\Big). 
\end{align}
The variance is a measure for the expected width of the scenario cone, therefore introducing such a dependency opens up the possibility to control the blow-up probability of the process without the necessity of including a state dependent mean reverting property.

\subsection{Existence and uniqueness}
The question of existence and uniqueness was answered only partially for one-dimensional stochastic drivers in \cite{MFLMM}. Since we aim to simulate more general multi-dimensional processes in the MF-LMM framework in our numerical study, we derive below general conditions for existence and uniqueness. 

Mean-Field SDEs, also called McKean-Vlasov SDEs or nonlinear diffusions, were introduced and studied in the works of McKean \cite{McKean66} and Sznitman \cite{Sni91}. In the classical set-up it is assumed that the dependency on the measure variable is linear: in the context of equation~\eqref{e:LMM-mf}, this would mean that $\sigma$ should be of the form $\sigma(t, \mu) = \int f(t, x)\,d\mu(x)$ for a function $f$ where we have suppressed the index. Since the variance is quadratic in the measure variable, the prescription \eqref{e:LMM-mf-var} does not fit into this scheme. Conditions for existence and uniqueness for mean-field SDEs with nonlinear measure dependency were derived by \cite{JMW08}, and the approach of this section is to use these results.  

Thus we need to show that $\sigma$ is Lipschitz in the measure variable. To do so, we use the $L$-derivative introduced by  P.-L.\ Lions (\cite{Car12}).

Before turning to the question of existence and uniqueness, we introduce the notion of an interacting particle system (IPS) associated to a mean-field equation. The IPS associated to \eqref{e:LMM}-\eqref{e:LMM-mf} is a system of SDEs that is defined as follows. Let $n\in\N$ and $1\le p\le n$. For a random variable $X$ denote by $\delta_X$ the Dirac measure centered at $X$.  The IPS associated to the mean-field Libor market model is
\begin{align}
 \label{e:IPS-LMM}
 d L^{i,p,n}_t
 &= L^{i,p,n}_t\sigma_i(t, \mu_t^{i,n})^{\bot}\,dW^{i,p} \\  
 \mu_t^{i,n}
 &= \frac{1}{n}\sum_{p=1}^n \delta_{L_t^{i,p,n}}
\end{align}
where $W^{i,p}$ are mutually independent $d$-dimensional Brownian motions. Note that $\mu_t^{i,n}$, which is called the empirical measure, is a stochastic measure. The initial conditions for $L_t^{i,p,n}$ are the same as those for $F_t^i$, that is, given by the initial forward rate $F_0^i$. 

Firstly, the IPS~\eqref{e:IPS-LMM} has structural relevance:  Due to the empirical measure, two processes $L^{i,p,n}$ and $L^{i,q,n}$ will in general not be independent. However, in the limit as $n$ goes to infinity one may hope that $L^{i,p,n}$ tends to a process $L^{i,p}$ which is a copy of the mean-field SDE \eqref{e:LMM}-\eqref{e:LMM-mf}. Thus the  processes  $L^{i,p,n}$ and $L^{i,q,n}$ would become independent in the limit $n\to\infty$, and this leads to the notion of propagation of chaos (see \cite{Sni91} and \cite{JW17} for a review). Secondly, for our purposes, the approximating IPS is not only of structural importance but also crucial from a practical point of view since the numerical simulation of a mean-field SDE is realized as a Monte-Carlo simulation of the approximating IPS, and this is the basis for the numerical algorithm in Section~\ref{sec:numALM}.  The corresponding result is provided in Theorem~\ref{thm:mflmm} below.

Let $\mathcal{P}_2(\R)$ be equipped with the Wasserstein distance 
\begin{equation}
 W_2(\mu_1, \mu_2)
 = \textup{inf}_{\pi\in\mathcal{C}(\mu_1,\mu_2)}
 \Big(
  \int_{\mathbb{R}^2} |x-y|^2\,\pi(dx,dy) 
 \Big)^{1/2}
\end{equation}
where $\mathcal{C}(\mu_1,\mu_2) = \{\gamma\in\mathcal{P}(\R\times\R) \textup{ with marginals } \mu_1\textup{ and }\mu_2\}$ denotes the set of all couplings. 
The space $\mathcal{P}_2(\R)$ is a Polish (i.e., separable and  completely metrizable topological) space. 

Since $\mathcal{P}_2(\R)$ is a nonlinear space, the notion of differentiability requires attention.
Consider a continuous function $f: \mathcal{P}_2(\R) \to\R$.  Following \cite{Car12,BRW21}, the derivative of $f$ at $\mu$ is defined by composing the function with a curve, $\mu_{\eps}$, through $\mu$ and evaluating $\partial f(\mu_{\eps})/\partial\eps|_{\eps=0}$.  
Concretely, let $\id: x\mapsto x$ be the identity map on $\R$. Given $\mu\in\mathcal{P}_2(\R)$ and 
$\phi\in L^2(\R,\mu)$, 
we define a curve with values in $\mathcal{P}_2(\R)$ by $\eps\mapsto\mu_{\eps}^{\phi} := (\id+\eps\phi)_{\sharp}\mu$. 

\begin{definition}[\cite{BRW21}]\label{e:defL}
Let $f: \mathcal{P}_2(\R) \to\R$ be a continuous function. 
\begin{enumerate}
\item
$f$ is called \emph{intrinsically differentiable} at $\mu\in\mathcal{P}_2(\R)$ with derivative $D^Lf(\mu)$ if
\[
L^2(\R,\mu)
\to\R,\quad
\phi\mapsto 
\frac{\partial}{\partial\eps}\Big|_{\eps=0} f(\mu_{\eps}^{\phi}) 
= 
\lim_{\eps\to0}
\frac{ f(\mu_{\eps}^{\phi} ) - f(\mu) }{\eps}
=: D^L_{\phi}f(\mu)
\]
is a bounded linear functional. In this case, $D^Lf(\mu)$ is characterized via the natural $L^2(\R,\mu)$-pairing through $\vv<D^Lf(\mu), \phi>_ {L^2(\R,\mu)}(\mathbb{R}) =  D^L_{\phi}f(\mu)$.
\item 
If additionally,
\[
 \lim_{||\phi||_{L^2(\R,\mu)}\to0}\frac{ | f((\id+\phi)_{\sharp}\mu) - f(\mu) - D^L_{\phi}f(\mu) | }{||\phi||_{L^2(\R,\mu)}}
 = 0
\]
for all $\phi\in L^2(\R,\mu)$, where $||\phi||_{L^2(\R,\mu)}$ is the $L^2$-norm of $\phi$, then $D^Lf(\mu)$ is called the \emph{Lions} or $L$-\emph{derivative} of $f$ at $\mu$.  
\end{enumerate}
\end{definition}

Consider functions $g$ and $h$ such that 
\begin{itemize}
\item 
$g: \R\to\R$ is differentiable;
\item 
$h: \R\to\R$ is twice differentiable, and $h\in L^1(\R,\mu)$, $h'\in L^2(\R,\mu)$ and $h''\in L^{\infty}(\R)$; 
\end{itemize}

\begin{remark}
These assumptions ensure that the calculation below make sense. Incidentally, we are interested in the cases $h(x)=x$ and $h(x)=x^2$, cf.~\eqref{e:LMM-mf-var}, whence these assumptions are not restrictive. On the other hand, it is relevant for our purposes that we do not need to assume more than the existence of the derivative of $g$.  Indeed, this allows to consider regime switching functions depending on $\max(\operatorname{Var}_{\mathcal{Q}^i}[F_t^{i}] - v_0, 0)$ where $\operatorname{Var}_{\mathcal{Q}^i}[F_t^{i}]$ is the variance of the process and $v_0$ is  a pre-defined variance threshold. Such a dependency is used in the dampening construction in \cite{MFLMM}, and in the combined dampening-anti-correlation implementation in Section~\ref{sec:esg}
\end{remark}

Define $f: \mathcal{P}_2(\R)\to\R$ through
\begin{equation}
 \label{e:ex1}
 f(\mu) = g(\mu(h))
\end{equation}
where $\mu(h) = E_{\mu}[h] = \int_{\mathbb{R}}h(x)\,\mu(dx)$. 
In the following, the notation $k(\eps) = o(\eps)$ means that an expression $k(\eps)$ satisfies $\lim_{\eps\to0}k(\eps)/\eps = 0$.
Observe that, for $\phi\in L^2(\R, \mu)$ and $\eps>0$, Taylor's expansion implies
\begin{align*}
 h(y+\eps\phi(y)) 
 = h(y) + h'(y)\eps\phi(y) + \frac{h''(\xi)}{2}\eps^2\phi(y)^2 
 = h(y) + h'(y)\eps\phi(y) + o(\eps)\phi(y)^2
\end{align*}
for a $\xi\in[y, y+\eps\phi(y)]$ and where $o(\eps)$ is independent of $y$ since $h''\in L^{\infty}(\R)$. 
Thus
\begin{align*}
 f\Big(
  \mu_{\eps}^{\phi} 
 \Big) 
 &= 
  g\Big( 
  \int_{\mathbb{R}} h(x)\,
  ((id+\eps \phi)_{\sharp}\mu)\,(dx) 
 \Big) 
 =
 g\Big( 
  \int_{\mathbb{R}} h(y+\eps\phi(y)) \,
  \mu\,(dy) 
 \Big) 
 \\ 
 &= 
  g\Big( 
  \int_{\mathbb{R}} 
   \Big(
    h(y) + \eps \phi(y)h'(y)  + o(\eps)\phi(y)^2
   \Big) \,\mu\,(dy) 
 \Big)  \\ 
 &= 
 g\Big(\mu(h)\Big)
 +  g'\Big(\mu(h)\Big)
  \Big(
   \eps \vv<\phi,h'>_{L^2(\R,\mu)} + o(\eps)||\phi(y)||_{L^2(\R,\mu)}^2
  \Big) 
 + o\Big(
   \eps \vv<\phi,h'>_{L^2(\R,\mu)} + o(\eps)||\phi(y)||_{L^2(\R,\mu)}^2
 \Big) \\ 
 &= 
 f\Big(\mu\Big)
 + 
 \eps \Big\langle g'(\mu(h)) h', \phi \Big\rangle_{L^2(\R,\mu)}
 + o\Big(\eps\Big) 
\end{align*}
Here we use that  
$
 o(\eps)g'(\mu(h))||\phi(y)||_{L^2(\R,\mu)}^2
 + 
  o(
   \eps \vv<\phi,h'>_{L^2(\R,\mu)} + o(\eps)||\phi(y)||_{L^2(\R,\mu)}^2
 )
 = o(\eps)
$.
Therefore, Definition~\ref{e:defL} yields 
\[
 D_{\phi}^Lf(\mu) = \vv< g'(\mu(h)) h', \phi>_{L^2(\R,\mu)}.
\] 
 
The Cauchy-Schwarz inequality now implies that the intrinsic derivative is also a Lions derivative and  is given by
\begin{equation}\label{e:ex1a}
 D^Lf(\mu)(x) 
 = g'(\mu(h)) h'(x) 
\end{equation}
for $x\in\R$. 
Slightly more generally, we can consider vector valued functions $f: \mathcal{P}_2(\R)\to\R^d$ of the form 
\begin{equation}
 f(\mu)
 = 
 \Big(   
  g^j( \mu(h_1), \dots, \mu(h_l))
 \Big)_{j=1}^d 
\end{equation}
where $g^j$, $j=1,\dots,d$, and $h_k$, $k=1,\dots,l$, are functions satisfying the same conditions as $g$ and $h$, respectively. Then we find that the Lions derivative $D^L f(\mu): \R\to\R^d$ at $\mu$ is given by  
\begin{equation}
\label{e:ex2}
 D^L f(\mu)
 = 
 \Big(
  \sum_{k=1}^l\partial_k g^j(\mu(h_1),\ldots,\mu(h_l))\,h_k'
 \Big)_{j=1}^d.
\end{equation}

The following lemma is provided for completeness' sake.

\begin{lemma}\label{lem:lip}
Assume $f: \mathcal{P}_2(\R)\to\R^d$ is continuous and that the Lions derivative $D^Lf(\mu)$ exists for all $\mu$ in $\mathcal{P}_2(\R)$. If $||D^Lf(\mu)||_{L^2(\mathbb{R}^d,\mu)} \le K$, for all $\mu$ in $\mathcal{P}_2(\R)$, then $f$ is Lipschitz continuous with Lipschitz constant $K$. 
\end{lemma}

\begin{proof}
Let $\widetilde{f}: L^2 = L^2(\Om,\mathcal{F},P)\to\R^d$ be an extension of $f$ defined by $\widetilde{f}(X) = f(\operatorname{Law}(X))$ where $(\Om,\mathcal{F},P)$ is a  probability space, see \cite{Car12}. We use \cite[Thm~6.2]{Car12} which yields $D^L f(\mu) = D\widetilde{f}(X)$, where $D$ is the Fr\'echet derivative, and $D\widetilde{f}(X)$ does not depend on the extension but only on the law $\mu = \operatorname{Law}(X)$. Here, the Fr\'echet derivative is identified with the corresponding element in $L^2$ via the Hilbert space pairing $E_P[\vv<X, Y>] = \vv<X,Y>_{L^2}$ where the angle bracket without subscript denotes the Euclidean inner product.
Consider elements $\mu_k = \operatorname{Law}(X_k)$ with $k=0,1$ in $\mathcal{P}_2(\R)$. 
Using the Cauchy-Schwarz inequality
\begin{align*}
 \Big|
  f(\mu_1) - f(\mu_0)
 \Big|
 &= 
 \Big|
 \int_0^1
  \frac{\partial}{\partial s}
  \widetilde{f}\Big(
   X_0+s(X_1-X_0)
  \Big) 
  \,ds
 \Big| 
 \\
 &= 
  \Big|
 \int_0^1
   E_{P}\Big[
   \Big\langle 
   D \widetilde{f}
   \Big(
    X_0+s(X_1-X_0)
   \Big) 
   , \Big(X_1-X_0\Big)
   \Big\rangle 
   \Big]
  \,ds
 \Big| 
 \\
 &= 
   \Big|
 \int_0^1
   \Big\langle 
   D^L f 
   \Big(
    \operatorname{Law}(X_0+s(X_1-X_0))
   \Big) ,
   \Big(X_1-X_0\Big)
   \Big\rangle_{L^2}
  \,ds
 \Big| 
 \\
 &\le 
 K \Big|\Big|X_1 - X_0\Big|\Big|_{L^2}.
\end{align*}
Since this holds for all extensions and since the Wasserstein distance satisfies 
$d_2(\mu_0,\mu_1) 
 \le  ||X_1 - X_0||_{L^2}$ 
when $\operatorname{Law}(X_0) = \mu_0$ and $\operatorname{Law}(X_1) = \mu_1$, the lemma follows.
\end{proof}

\begin{theorem}[Existence, uniqueness and IPS approximation of MF-LMM]\label{thm:mflmm}
Consider the mean-field Libor market model be given by \eqref{e:LMM}-\eqref{e:LMM-mf}. Let $\lam_i = (\lam_i^j)_{j=1}^d: [0,t_i]\times\R^l\to\R^d$ be Lipschitz continuous in $t$ and $L$-differentiable in $\mu$ such that  
\[
 D^L\lam_i(t, \mu)
 = 
  \Big(
  \sum_{k=1}^l
   \partial_{k+1} \, \lam_i^j(t, \mu(h_1),\ldots,\mu(h_l))\,h_k'
 \Big)_{j=1}^d
\]
is bounded in $L^2(\R^d,\mu)$, uniformly for all $t\in\R_+$ and all $\mu\in\mathcal{P}_2(\R)$. If \eqref{e:LMM-mf} is of the form \eqref{e:LMM-mf-var}, then this is the case if $\lam_i: [0,t_i]\times\R_+\to\R^d$, $(t,v)\mapsto\lam_i(t,v)$ is Lipschitz continuous in $t$, differentiable in $v$, and satisfies the pointwise bound  
\begin{equation}
\label{e:mflmm_p}
v\sum_{j=1}^d|\partial_2\,\lam_i^j(t,v)|^2 \le K
\end{equation}
with respect to a constant $K<\infty$ which is independent of $(t, v)\in[0,t_i]\times\R_+$. 
\begin{enumerate}
\item 
Then the system has a unique solution $F_t^i$ such that $E[ \sup_{0\le t\le t_i} |F_t^i|^2] < \infty$.
\item 
Moreover, pathwise propagation of chaos holds: 
\[
 \lim_{n\to\infty}\,\sup_{p\le n}\,
 E\Big[
  \sup_{t\le t_i}\, \Big| L_t^{i,p,n} - L_t^{i,p} \Big|^2
 \Big]
 = 0
\] 
where $L^{i,p,n}$ is the approximating IPS \eqref{e:IPS-LMM} and $L^{i,p}$ is a version of \eqref{e:LMM}-\eqref{e:LMM-mf} defined with respect to $W^{i,p}$.
\end{enumerate}
\end{theorem}

\begin{proof}
The result follows directly from \cite[Prop.~2 and Thm.~3]{JMW08} if it can be shown that $\sigma_i$ is Lipschitz continuous also in the measure variable. In view of Lemma~\ref{lem:lip}, this, in turn, holds if the Lions derivative $D^L\sigma_i(t, \mu)$ exists and is bounded in $L^2(\R^d,\mu)$ with respect to some constant independent of  $t$ and $\mu$.  Given the system \eqref{e:LMM}-\eqref{e:LMM-mf}, the Lions derivative is of the form \eqref{e:ex2}. Hence the result follows. In the special case where the system is given by \eqref{e:LMM}-\eqref{e:LMM-mf-var} we have  $h_1(x) = x$ and $h_2(x) = x^2$. 
For better readability we omit the index $i$ and the time dependency from now on.  Owing again to \eqref{e:ex2},
\[
 \sigma(\mu) = \lam(\mu(h_2) - \mu(h_1)^2)
\]
has an $L$-derivative $D^L\sigma$ whose $L^2(\R^d,\mu)$-norm is 
\begin{align*}
 \Big|\Big|
  D^L\sigma (\mu)
 \Big|\Big|^2_{L^2(\mathbb{R}^d,\mu)}
 &= 
 \sum_{j=1}^d\Big|\Big|
   -2\mu(h_1)(\lam^j)'\Big(\mu(h_2) - \mu(h_1)^2\Big) h_1' + (\lam^j)'\Big(\mu(h_2) - \mu(h_1)^2\Big) h_2'
 \Big|\Big|^2_{L^2(\mathbb{R}^d,\mu)}
 \\ 
 &= 
 4 \sum_{j=1}^d \Big|(\lam^j)'(v)\Big|^2
 \Big|\Big|
  h_1 - \mu(h_1)
 \Big|\Big|^2_{L^2(\mathbb{R}^d,\mu)} \\ 
 &= 
 4 v \sum_{j=1}^d \Big|(\lam^j)'(v)\Big|^2
\end{align*}
where $v = \mu(h_2) - \mu(h_1)^2$ is the variance of $\mu$. Hence  boundedness of 
$||  D^L\sigma (\mu) ||^2_{L^2(\mathbb{R}^d,\mu)}$ 
is ensured by the pointwise assumption.
\end{proof}

\begin{remark}
The MF-LMM was introduced in \cite{MFLMM} where it was shown in a numerical study that the mean-field interaction can lead to a reduction of blow-up probability, i.e.\ a reduction in the number of scenarios where the forwards exceed a pre-defined threshold. Such a reduction of blow-up probability is desirable for numerical reasons as mentioned in the introduction to this section. However, it is also a feature that is in agreement with the postulate to model a \emph{realistic} forward dynamic: exploding interest rates are not observed in reality, and central banks act to stabilize rates whence such an explosion is not only not observed but also very implausible. It should be stressed that the mean-field interaction encoded in \eqref{e:LMM-mf-var} is not interpreted as an interaction between different forward rate scenarios. Such an interpretation would not make sense since different unfoldings of reality cannot communicate with each other. Rather, the underlying idea is that of an ulterior agency (i.e., central bank policy) which controls yield curves to the effect that the variance of possible yield movements is restricted. It is this effect that is modelled by the mean-field dependency. 
\end{remark}

\subsection{Projection along tenor dates}
Having shown existence and uniqueness, the mean-field Libor market model may be regarded as a generic Libor market model with time-dependent volatility structure. In particular, the transformation to the spot measure may be carried out. However, as shown in \cite{MFLMM} the transformed system is in general not a mean-field system, unless one introduces an auxiliary process. For our envisaged purposes of performing long term simulations it is standard to consider yearly time steps along tenor dates. Along these dates the picture simplifies and the transformed system is, in fact, a mean-field SDE with respect to the spot measure. This is relevant since it implies that the approximating IPS can be used to define a Monte-Carlo routine based on the associated Euler-Maruyama scheme.  

Let 
\begin{equation}
B(t_j)
= (1+\delta F_{t_{j-1}}^{j-1})B(t_{j-1}),
\quad
B(t_0) = 1\,,
\end{equation}
denote the implied money market account (i.e., the numeraire) and $\mathcal{Q}^*$ the associated spot measure. 
If $t=t_j$ is a tenor date, then the auxiliary process, alluded to above, can be expressed as
\begin{equation}
Y_{t_j}^m
=
B(t_j)^{-1}\frac{P(j,m)}{P(0,m)}\,,
\end{equation}
where 
\begin{equation}\label{eq:bond}
P(j,m)
=
\Pi_{l=j}^{m-1}(1+\delta_l F_{t_j}^l)^{-1}\,,
\end{equation}
is the time $t_j$-value of one unit of currency paid at $t_m$. 
With
\begin{equation}
\label{e:phi1}
\Psi_j^{m}
:= 
E_{\mathcal{Q}^*}\left[\left(
F_{t_j}^m - E_{\mathcal{Q}^*}\left[F_{t_j}^m  B(t_j)^{-1}\frac{P(j,m)}{P(0,m)}\right]
\right)^2  B(t_j)^{-1}\frac{P(j,m)}{P(0,m)} \right]\,
\end{equation}    
it follows that the evolution along the tenor dates of the mean-field Libor market model with respect to the spot measure is given  by
\begin{equation}
\label{e:spot_phi}
d F_{t_j}^m
= 
F_{t_j}^m\left(
\sum_{k=j+1}^m\frac{\delta F_{t_j}^k}{\delta F_{t_j}^k + 1}
\lambda^k(t_j,\Psi_j^{k})^{\top}
\lambda^m(t_j,\Psi_j^{m}) \di t
+
\lambda^m(t_j,\Psi_j^{m})^{\top} \mathrm{d} W^*_t\,
\right),
\end{equation}
since $\eta(t_j)=j+1$. See~\cite{MFLMM}.

Thus it remains to specify the functional form of $\lambda^i: [0,t_{i}]\times\R_+\to\R^d$ which has to be Lipschitz continuous in $t$ and should satisfy the pointwise condition~\eqref{e:mflmm_p}. 
A numerical study concerning various such choices is provided in \cite{MFLMM}, where it is also shown numerically that, in the long run, the classical Libor market model can lead to exploding rates while the mean-field controlled models all significantly reduce blow-up probability. 

\section{Life insurance with profit participation} \label{sec:li}

Traditional life insurance with profit participation is characterized by a strong interdependence between assets and liabilities. On the liability side, the policyholder contracts are equipped with a minimum guarantee rate and participate in the company's profit. This profit depends on the book value return of the assets under management. The company's management may control -to a certain extent- the book value return and the crediting strategy, and this control typically depends on information from liabilities (e.g., guaranteed interest rate) and assets (e.g., amount of unrealized gains). The purpose of this section is to describe these allocation processes and the corresponding control parameters in a mathematically concise manner. 

\subsection{The portfolio view}
The liabilities of a traditional life insurance consist of two items, namely the life provisions, $LP$, and the surplus fund, $SF$. Thus the book value of liabilities is
\[
 BV = LP + SF. 
\]
In this picture we have assumed that there is no free capital. The life provisions are attributed to individual contracts while the surplus fund belongs to the collective of policyholders as a whole. The former can be further split into a mathematical reserve, $V$, and declared bonuses (or benefits), $DB$, such that $LP = V + DB$. Thus $V$ denotes the sum of all contracts' mathematical reserves calculated according to actuarial principles and $DB$ is likewise the sum of all declared benefits. 

Let $\mathcal{A}$ be the set of all assets which are dedicated to the profit participating business. We denote book and market values associated to a particular asset $a$ by $BV^a$ and $MV^a$, respectively. The unrealized gains are by definition
\begin{equation}
 UG^a = MV^a - BV^a
\end{equation}
and these can be positive or negative. In the  latter case we will also refer to $UG$ as unrealized losses. 
Balance sheet parity dictates that we must have 
\begin{equation}
 BV = \sum_{a\in\mathcal{A}} BV^a .
\end{equation}
In reality, the right hand side of this equation would be expected to be slightly larger than the left hand side. However, since profit participation is calculated only with respect to those assets balancing liabilities, exact balance sheet parity should be assumed. 

Let from now on $t = 0,\dots,T$ denote yearly time steps where $T$ is the projection horizon, and $T$ will correspond to $60$ years or even more since life insurance is a long term business. For a time dependent quantity $f_t$  we denote the increment by $\Delta f_t = f_t-f_{t-1}$. 

The return on assets generated by the portfolio at time $t$ is 
\begin{equation}
 ROA_t 
 = \sum_{a\in\mathcal{A}_{t-1}}\Delta BV_t^a
 = \sum_{a\in\mathcal{A}_{t-1}}(BV_t^a - BV_{t-1}^a) 
\end{equation}
where $\mathcal{A}_{t-1}$ are all assets under management at the previous time step. Notice that $ROA$ is defined in terms of book values. This terminology requires some explanation: Consider an asset $a\in\mathcal{A}_{t-1}$ with book value $BV_{t-1}^a$ and unrealized gains $UG_{t-1}^a = MV_{t-1}^a - BV_{t-1}^a$. Assume these values change due to some mechanism (e.g., market movements or depreciation rules) at $t$ to $(BV_t^a)'$ and $(UG_t^a)'$. Then the \emph{uncontrolled} return due to $a$ is $(BV_t^a)' - BV_{t-1}^a$. However, management may decide to sell $a$ at $t$ and thereby realize the positive or negative value $(UG_t^a)'$. In fact, we will allow partial selling such that management may choose $u_t^a\in[0,1]$ at $t$, whereby the final book value becomes $BV_t^a = (BV_t^a)' + u_t^a(UG_t^a)'$. Similarly, the unrealized gains at $t$ can now be expressed as $UG_t^a = (1-u_t^a)(UG_t^a)'$. 
The corresponding return is thus 
\[
 ROA_t^a = \Delta BV_t^a = (BV_t^a)' + u_t^a(UG_t^a)' - BV_{t-1}^a, 
\]
and this depends on the changes in the book value and in the unrealized gains, and on the applied management rule. 

The life provisions, $LP_t$, can be decomposed as the sum $LP_t = V_t+DB_t^0+DB_t$: here $V_t$ denotes the sum of mathematical reserves at $t$, $DB_t^0$ are declared bonuses at $t$ which were declared before and up to valuation time $t=0$, and $DB_t$ are declared bonuses at $t$ which were declared at times $1\le s\le t-1$. The mathematical reserve of a contract is the reserve that is calculated according to classical actuarial principles including safety loadings (\cite{Gerber}). Since $DB_t^0$ is guaranteed at time $0$ we denote the guaranteed life provisions by $LPG_t = V_t + DB_t^0$. The declared bonuses, $DB_t$, are increased by declarations, $ph_t^*$, and decreased due to benefits, $ph_t$, and surrender gains, $sg_t^*$. That is, 
\begin{equation}
\label{e:DBevol}
 \Delta DB_t = ph_t^* - ph_t - sg_t^*. 
\end{equation}

The company experiences the following cash flows on the liability side: premiums, benefits and expenses. Premiums paid at $t$ are denoted by $pr_t$. Benefits are split into guaranteed benefits, $gbf_t$, stemming from $LPG_{t-1}$ and  discretionary benefits, $ph_t$, stemming from $DB_{t-1}$.
That is, the total benefit cash flow at $t$ is $bf_t = gbf_t+ph_t$.
Moreover, the guaranteed benefits can be further decomposed as $gbf_t = abf_t - \Delta DB_t^0$ where $abf_t$ denotes the benefit calculated according to actuarial principles (\cite{Gerber}) and $\Delta DB_t^0$ is the reduction of $DB_t^0$. Note that the latter can only decrease since new profits are declared to $DB_t$ by construction. Finally, the sum of expense cash flows at $t$ is denoted by $exp_t$. 

The gross surplus at $t$, which corresponds to the company's profit under locally generally accepted accounting principles (local GAAP), is now defined as
\begin{equation}
\label{e:gs} 
 gs_t 
 = ROA_t - \Delta LPG_t + pr_t - gbf_t - exp_t + sg_t^*
 = ROA_t - \Delta V_t + pr_t - abf_t - exp_t + sg_t^*.
\end{equation}
Let $\rho_{t}$ be the weighted average technical interest rate at $t$. That is, $\rho_{t} = \sum_{c \in\mathcal{L}_{t-1}} \rho_c V_{t-1}^c / V_{t-1}$ where  $\mathcal{L}_{t-1}$ denotes the liability portfolio at $t-1$ and $c \in\mathcal{L}_{t-1}$ is a contract with guarantee rate $\rho_c$ and mathematical reserve $V_{t-1}^c$. If the premium were calculated without any safety loadings, we would have $\Delta V^c_t - pr^c_t + abf^c_t + exp^c_t = \rho^c V_{t-1}^c$. Thus the difference $\rho_{t}V_{t-1} - \Delta V_t + pr_t - abf_t - exp_t + sg_t^* =: \gamma_{t}V_{t-1}$ can be viewed as a technical gain due to factor loadings, and we take this equation as the definition for the technical gains rate 
\begin{equation}
\label{e:gamma_def}
 \gamma_{t}
 = 
 \Big(\rho_{t}V_{t-1} - \Delta V_t + pr_t - abf_t - exp_t + sg_t^*
 \Big) / V_{t-1}
\end{equation}
The gross surplus equation can now be expressed succinctly as
\begin{equation}
\label{e:gs_g} 
 gs_t 
 = ROA_t - (\rho_t-\gamma_t)V_{t-1} .
\end{equation}

If the gross surplus is negative, the shareholder has to cover this and incurs a cost of guarantee. Thus we define $cog_t = gs_t^-$ where $k^- = -\min(k,0)$ is the negative part of a real number. If the gross surplus is positive, this profit is shared between policyholder, shareholder and tax office according to positive factors $gph$, $gsh$ and $gtax$ satisfying $gph+gsh+gtax = 1$. Let $k^+ = \max(k,0)$ be the positive part of a real number. 
Thus, at $t$, the shareholder receives a cash flow  $shg_t = gsh\cdot gs_t^+$ representing shareholder gains, and the tax office receives $tax_t = gtax\cdot gs_t^+$. However, the policyholder collective need not be credited the full amount $gph\cdot gs_t^+$ at $t$. Rather, management may choose a control parameter $\nu_t\in[0,1]$ such that $\nu_t\cdot gph\cdot gs_t^+$ is credited at time $t$ and then paid out at later times $s>t$ when the corresponding policies which were credited at $t$ reach their maturities. The remaining fraction $(1-\nu_t)\cdot gph\cdot gs_t^+$ is allocated to the surplus fund $SF$. Conversely, the declaration may be, again at the  management's discretion, augmented by choosing $\eta_t\in[0,1]$ and declaring the amount $\eta_t\cdot SF_{t-1}$. The total declaration at time $t$ is therefore
\begin{equation}
\label{e:cred}
 ph_t^*
 = 
 \nu_t\cdot gph\cdot gs_t^+ + \eta_t\cdot SF_{t-1}. 
\end{equation}
We define $\mu_s^t$ as the fraction of the declaration $ph_t^*$ that is paid out at $s>t$. Notice that $\mu_s^t$ is a model dependent quantity. 
 
It follows that $ph_t^*$ gives rise to delayed cash flows $\mu^t_{t+1}ph_t^*, \ldots, \mu^t_{T}ph_t^*$ and that the policyholder cash flow can in turn be expressed as 
\begin{equation}
\label{e:mu}
 ph_t = \sum_{s=1}^{t-1} \mu_t^s ph_s^*.  
\end{equation}
A schematic overview of the declaration procedure, including the management controls which are denoted by $u$, $\nu$ and $\eta$,  is contained in the diagram in Figure~\ref{d:gs}.  
\begin{figure}[ht]
\begin{equation*}
\xymatrix{
      &  &  & SF\ar@{-->}[d]^{\textcolor{red}{\eta}} &  \\
    \textup{Liabilities} \ar[dr]        &  &
    gph\cdot gs^+ \ar[r]_{\phantom{xxx}\textcolor{red}{\nu}}
    \ar[ur]^{1-\textcolor{red}{\nu}}
    & ph^* \ar@{~>}[r]^{delay} 
    & \textcolor{blue}{ph} \\ 
    & gs\ar[ur]^{}\ar[r]^{}\ar[dr]_{} &  gsh\cdot gs^+ \ar@{=}[r]
    &\textcolor{blue}{shg}   &  \\ 
    \textup{Assets}\ar[ur]^{ROA(\textcolor{red}{u})}       
    & \phantom{xxxx}\textcolor{blue}{cog}=gs^{-}\ar@{-->}[u] &  gtax\cdot gs^+ \ar@{=}[r] & \textcolor{blue}{tax}  &  
    }
\end{equation*}
\caption{Schematic depiction of the declaration process. Management rules ($u$, $\nu$, $\eta$) are marked red, cash-flows ($cog$, $tax$, $shg$, $ph$) are blue. Strategic asset allocation is not included in the diagram.}
\label{d:gs}
\end{figure}
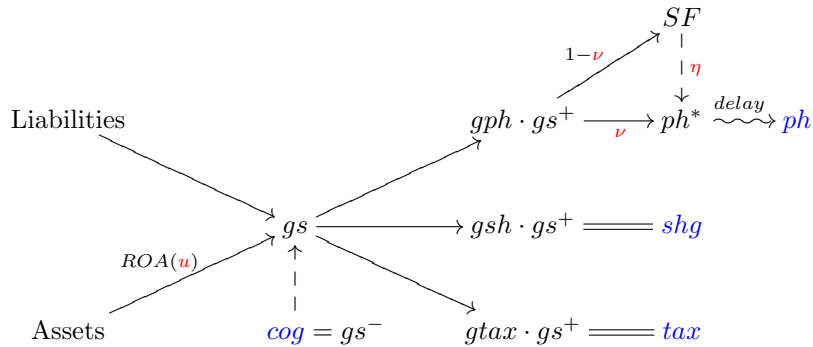


\subsection{Evolution of policyholder profits}
The evolution equation~\eqref{e:DBevol} can be adjoined by one for the surplus fund, namely
\begin{equation}
 \label{e:SF} 
 \Delta SF_t 
 = (1-\nu_t)gph\cdot gs_t^+ - \eta_t SF_{t-1} 
 = gph\cdot gs_t^+ - ph_t^*.
\end{equation}
The sum $DB_t + SF_t$ represents the pot of shared profits incurred up to time $t$. Its evolution can be expressed  as 
\begin{equation}\label{e:evol_pot} 
\Delta(DB_t + SF_t) 
= gph\cdot gs_t^+ - ph_t - sg_t^*
\end{equation}
which, crucially, holds independently of declaration management rules regarding $\nu$ and $\eta$.   

Let $MV_t = \sum_{a\in\mathcal{A}_t}MV_t^a$ denote the market value of assets at $t$. Further, let $B_t^{-1}$ denote the stochastic discount factor associated to a risk neutral interest rate model, and let $E[\_]$ be the expectation associated to the risk neutral measure. In Section~\ref{sec:numALM} this measure is fixed as the spot measure associated to the mean-field Libor market model of Section~\ref{sec:mflmm} but for the present purpose the specific choice is not relevant. 
 
The value of guaranteed benefits is defined as 
\begin{equation} 
\label{e:GB}
 GB = E\Big[\sum_{t=1}^T B_t^{-1}(gbf_t + exp_t - pr_t)\Big]
\end{equation}
and the value of future discretionary benefits is 
\begin{equation}
\label{e:FDB} 
FDB
 = E\Big[\sum_{t=1}^T B_t^{-1}ph_t\Big].
\end{equation}
Further, we consider the value of shareholder gains, $SHG = E[\sum_{t=1}^T B_t^{-1}shg_t]$, the shareholder's cost of guarantee
\begin{equation}\label{e:cog}
 COG = E[\sum_{t=1}^T B_t^{-1}cog_t],
\end{equation}
and the value of tax payments, $TAX = E[\sum_{t=1}^T B_t^{-1}tax_t]$. 


The no-leakage principle (\cite{HG19}) states  that 
\begin{equation}
\label{e:noL}
 MV_0 = GB + FDB + SHG - COG + TAX + E[B_T^{-1}MV_T] . 
\end{equation}

\subsection{Balance sheet representation of future discretionary benefits} 
Putting the above together yields a representation which can be stated without any cash flows: 

\begin{theorem}[\cite{GH22}]\label{thm:rep}
\begin{equation}\label{e:fdb-rep}
    FDB 
    = 
    FDB_0 
    + gph\cdot COG - I - II - III 
\end{equation}
where
\begin{align*}
    FDB_0 
    &=
    SF_0 + gph\Big(LP_0 + UG_0 - GB \Big)  \\
    I
    &:= E\Big[B_T^{-1}\Big(DB_T+SF_T + gph(UG_T + V_T + DB_T^{0}) \Big)\Big] \\
    II
    &:= 
    (1-gph)E\Big[ \sum_{t=2}^T B_t^{-1} sg_t^*\Big] \\
    III
    &:=
    (1-gph) E\Big[ \sum_{t=1}^T F_{t-1} B_t^{-1}  (DB_{t-1} + SF_{t-1}) \Big] 
\end{align*}
\end{theorem}

\subsection{Optimal control}\label{sec:con}
The management's objective is, according to corporate practice, to maximize the shareholder value of in-force business, $VIF = SHG - COG$.  In view of the no-leakage principle~\eqref{e:noL} this means that management rules should be designed so as to minimize 
\[
 FDB + TAX + E[B_T^{-1}MV_T] \longto \textup{MIN} 
\]
since $MV_0$ and $GB$ cannot be altered by management rules (assuming independence of guaranteed benefits and profit sharing, i.e.\ static policyholder behaviour). 

Focusing on the $FDB$-part in the minimization problem, if we choose $\nu_t = \eta_t = 0$ for all $t = 1,\ldots,T$ it follows that $ph_t^*=0$ and the future discretionary benefits are minimized as $FDB=0$. Clearly, this strategy is unrealistic. Thus we have to reformulate the control problem $FDB + TAX + E[B_T^{-1}MV_T] \longto \textup{MIN}$ subject to certain constraints which reflect the management's goal of providing competitive declared bonuses to the policyholders. Typical constraints to this effect are: 
\begin{enumerate}
\item 
There is a fixed $\theta>0$ such that $SF_t\le \theta LP_t$. This is a run-off assumption ensuring that relation between different balance sheet items remains realistic throughout the projection and it automatically reduces the terminal value $E[B_T^{-1}MV_T]$. 
\item 
The profit share $ph_t^*$ aims to follow a target reflecting the policyholders' expectations. This target could be defined, e.g., in terms of a combination of the prevailing yield curve at $t$ to reflect the expectation of a financially rational policyholder, and of the previous profit sharing rate to avoid large jumps in the profit sharing rate.   
\end{enumerate}
A concrete implementation of such a set of rules is presented in Section~\ref{sec:MR}.

Another way to view this control problem goes as follows. The numbers $\mu_t^s$  defined in \eqref{e:mu} allow us to express the future discretionary benefits as 
\begin{align*}
 FDB 
 &= E\Big[ \sum_{t=2}^T B_t^{-1}\sum_{s=1}^{t-1}\mu_t^s \,ph_s^* \Big]
 = E\Big[\sum_{t=1}^{T-1} ph_t^*\, q_t\Big]
\end{align*}
where $q_{t} := E[\sum_{s=t+1}^T B_s^{-1}\mu^t_s|\mathcal{F}_t]$. Notice that $q_t\ge0$ and $ph_t^*\ge0$ by construction. This formulation allows to directly discern the impact of $ph_t^*$ on the $FDB$-minimization. If we assume that the strategic asset allocation (SAA) is not a control parameter but  a constraint, namely that the SAA should be kept approximately constant, the remaining actuated variables are $\nu$ and $\eta$ in \eqref{e:cred} and the realizations of unrealized gains in the calculation of $ROA$. The latter can be expressed as $u_t = (u_t^a)$ where $a\in\mathcal{A}_t$ runs over all assets held at $t$ and $u_t^a\in [0,1]$ controls the fraction of $a$ to be sold at $t$. 

Let $ta_t$ denote an appropriate target at $t$, defined e.g.\ as in item (2) above.  
Thus we may rephrase the control problem as the search for a strategy $(\nu_t,\eta_t,u_t)$ such that 
\begin{equation}
 E\Big[
  \sum_{t=1}^{T-1} \Big( ph_t^* \,q_t
   +   \om( ph_t^* - ta_t)^2
   \Big)
 \Big] 
 \longto\textup{MIN}
\end{equation}
where  $\om$ is a weighting  factor representing the target's importance and
\[ 
 ph_t^* = ph_t^*(\nu_t,\eta_t,u_t)
 = \nu_t \,gph\,\Big( ROA_t' + \sum_{a\in\mathcal{A}_t} u_t^a UG_t^a - (\rho_{t-1}-\gamma_{t-1})V_{t-1}
 \Big)^+
  + \eta_t\, SF_{t-1} ,
\]
and the minimization is subject to the constraints  $SF_t\le \theta LP_t$ with a fixed number $\theta>0$ and that the  SAA is kept (approximately) constant. Here, $ROA_t' = BV_t' - BV_{t-1}$ denotes the change in the portfolio's book value before application of management rules (i.e., realizations of unrealized gains). 

In practice, the admissibility of controls $u^a$ is often restricted further: while it is advantageous from the shareholder's perspective to realize unrealized losses in cases where $ph_t^* > ta_t$, realizations of positive unrealized gains in cases where $ph_t^* < ta_t$ might be prohibited to avoid future cost of guarantee payments. Indeed, when the gross surplus is negative it is certainly desirable, again from the shareholder's point of view, to realize positive unrealized gains in order to avoid capital injections.

\section{Numerical ALM modelling}\label{sec:numALM}

To calculate the present value, $FDB$, of future discretionary benefits in traditional life insurance, it is necessary to model the statutory balance sheet and to carry out a Monte-Carlo simulation along economic scenarios.

Consider yearly time steps $t=1,\dots,T$ where $T$ corresponds to the projection horizon. 
The cash-flows,  at time $t$, relevant for market consistent valuation (best estimate calculation according to Solvency~II) are premiums, $pr_t$, expenses, $exp_t$, guaranteed policyholder benefits, $gbf_t$,  and discretionary policyholder benefits, $ph_t$, cf.\ Section~\ref{sec:li}. 

The \emph{best estimate}, $BE$, is given as the  sum of guaranteed benefits and future discretionary benefits defined in \eqref{e:GB} and \eqref{e:FDB}, that is 
\begin{equation}
\label{e:BE}  
 BE = GB + FDB .
\end{equation}

This section contains a description of the Asset Liability Management (ALM) model which we have implemented in the programming language $R$ in order to numerically calculate $GB$ and $FDB$. 
The algorithm for $FDB$ is quite involved due to the interaction of asset module (Section~\ref{s:AM}), liability module (Section~\ref{s:LM}) and management rules (Section~\ref{sec:MR}). The general mechanism behind profit participation in traditional life insurance is described in Section~\ref{sec:li} and sketched at the level of model points in Section~\ref{s:PP}. 
 
\subsection{Contracts and model points}\label{s:PP}
Our ALM model is concerned with classical life insurance contracts. Each contract (or, more generally, model point), $x$, has a specified maturity payment, $M^x$, which is guaranteed and may participate in the company's earnings. More precisely, each contract $x$ 
\begin{itemize}
\item 
has a specific maturity time $T^x$;
\item 
has associated best estimate mortality and surrender tables;
\item 
has a constant technical interest rate $\rho^x$; 
\item 
pays a constant premium $pr^x$ up to $T^x-1$;  
\item 
has an associated mathematical reserve $V_t^x$ calculated according to classical actuarial assumptions involving, e.g., $\rho^x$;  
\item 
has an associated bonus account of total declared benefits $TDB_t^x$ depending on the company's profits; 
\item 
receives either a maturity benefit $M^x+TDB^x_{T^x-1}$ at $T_x$; or a death benefit $M^x+TDB^x_{t-1}$ at $t<T_x$; or, in case of surrender at $t<T_x$, a surrender benefit $\kappa_t (V_{t-1}^x + TDB_{t-1}^x)$ where $\kappa_t$ is a penalty term which is linear in $t$ such that $\kappa_0 = 0.9$ and $\kappa_{T^x}=1$.  
\end{itemize}

When we model the projection of a model point $x$  which has started before valuation time $t=0$, it is necessary to distinguish between past and (uncertain) future. Thus we denote by $(DB^x)^{0}_0$ those declared benefits associated to $x$ which have been declared in the past, i.e.\ prior to valuation time $t=0$. The modelled projection at $t$ of this account is denoted by  $(DB^x)^{0}_t$. Similarly, future (a priori, uncertain) declared benefits are denoted by $DB_t^x$ with $DB_0^x=0$. This leads to the decomposition $TDB_t^x = (DB^x)^{0}_t + DB^x_t$ of the total account.

\subsection{Economic Scenario Generator (ESG)}\label{sec:esg}
The proposed ALM model uses the mean-field Libor market model of Section~\ref{sec:mflmm} to generate market and book values for four different asset classes: cash, bonds (without default), equity and property. 

Cash is modelled as a bond with maturity of one year. This corresponds to the roll-over definition of the implied money market account. 

Market values of property and equity are assumed to follow a geometric Brownian motion where the drift depends on the prevailing one-year forward interest rate, $F_t$. Further, the drift is modelled as a fixed rate, $d$, representing rental income or dividend yield. Finally, each property or equity asset may have its own fixed volatility, $\sigma$. Hence the market value, $MV_t^a$ of a given property or equity asset is assumed to be projected according to 
\begin{equation}
\label{e:GBM}
 MV_t^a 
 = 
 MV_0^a \exp\Big( (F_t-d-\sigma^2/2)t + \sigma W_t^a \Big)
\end{equation}
where $MV_0^a$ is the initial market value and $W_t^a$ is the asset's Brownian motion whose correlation with stochastic drivers of other assets remains to be specified. Equity, property and bond portfolio are assumed to be independent in our model. We emphasize that $d$ and $\sigma$ are assumed to be fixed numbers while $F_t$ is realized as a numerical implementation of the MF-LMM. 

The numerical implementation we have chosen for this task is a superposition of the mean-field taming \cite[Section~4.3]{MFLMM} and the anti-correlation approach \cite[Section~4.5]{MFLMM}; this approach is chosen since it allows flexibility in reproducing cap and swaption prices while at the same time effectively reduces the probability of blow-up.

Let $\Psi^m$ be defined by \eqref{e:phi1}. Let $v_0>0$ denote a predefined variance threshold and $f\ge1$ a stretching of this  threshold. Thus when the  variance exceeds the threshold, $\Psi^m > v_0$, the dynamic enters a new regime (anti-correlation), and when $\Psi^m > fv_0$ the dynamic is supposed to be further affected by a volatility dampening factor. 
We consider the evolution equation: 
\begin{align} 
\label{e:anti-cor}
\lambda^m\Big(t,\Psi_t^{m}\Big)
&=
S(v_0-\Psi_t^m)\cdot r_m(t)
+ 
(1-S(v_0-\Psi_t^m))\cdot 
 \exp\Big(-\max(\Psi_t^{m}/ (fv_0) - 1,0)  \Big)\,
 \Big|r_m(t)\Big| e_n,\\
 \textup{if }
 m &= 2n-1
\\
\lambda^m\Big(t,\Psi_t^{m}\Big)
&=
S(v_0-\Psi_t^m)\cdot r_m(t)
- 
(1-S(v_0-\Psi_t^m))\cdot 
 \exp\Big(-\max(\Psi_t^{m}/ (fv_0) - 1,0)  \Big)\,
 \Big|r_m(t)\Big| e_n,\\
 \textup{if }
 m &= 2n
\\
r_m(t) 
&= 
\left(\Big(a(t_{m-1} - t) + d\Big) e^{-b(t_{m-1} -t)} + c\right)
\left(
\begin{matrix}
\cos\theta_m\\
\sin\theta_m
\end{matrix}
\right), 
\quad t \le t_{m-1} \,,
\end{align}
where $S$ is a sigmoid function with inflection point $0$ such that $\lim_{x\to\infty}S(x) = 1$  and $\lim_{x\to-\infty}S(x) = 0$. Numerically, this is approximated by a switch at $0$. 
Further, $n=1,\dots,N/2$ (we assume that $N$ is even) and $e_n$ is the $n$-th standard basis vector.
In this case the dimension of the Brownian increment in \eqref{e:spot_phi}, i.e.\ in \eqref{e:LMM-mf}, is $d=N$. 

The parameters which are used to calibrate the model are the so-called angles $\theta_m$, the Rebonato (\cite{RMW}) parameters $a$, $b$, $c$, $d$, and the mean-field parameters $f\ge1$ and $v_0>0$.  

The interest rate scenarios are generated by applying an Euler-Maruyama scheme to the approximating interacting particle system (IPS), and our ALM model uses 5000 such scenarios. 

Notice that the use of \eqref{e:anti-cor} and the associated Monte-Carlo simulation of the IPS  are justified by Theorem~\ref{thm:mflmm} since the volatility structure is of the form \eqref{e:LMM-mf-var} and the map $\psi\mapsto \exp(-\max(\psi/ (fv_0) - 1,0))$ clearly satisfies the point-wise condition~\eqref{e:mflmm_p} for all choices  $f\ge1$ and $v_0>0$; it suffices that the derivative of this map exists and is asymptotically $o(1/\psi)$ as $\psi\to\infty$ but it does not need to be continuous.

\subsection{Asset module}\label{s:AM}
The purpose of the asset module is to model book and market values of assets under management and to provide re-\ and deinvestment strategies. The latter are part of the management rules and described in Section~\ref{sec:MR}. 

The ALM model provides four different asset classes: cash, bonds, equity, and property. Bonds are assumed to be default-free, thus there is no distinction between corporate and government bonds. 

Cash is equivalent to a bond with a maturity of $1$, and correspondingly the interest earned by cash is the prevailing one-year forward rate. Further, book and market values, $BV^c$ and $MV^c$, for cash coincide.  That is,
\[
BV_t^c = MV_t^c = (1+F_{t-1})MV_{t-1}^c. 
\]
A bond, $b$, consists of a maturity $T^b$, a nominal payment $N^b$, a coupon factor $K^b$, such that  the coupon payment is $K^bN^b$, a market value $MV^b$ and a book value $BV^b$. At each time step $t<T^b$ the market value is determined by the prevailing yield curve, i.e.
\[
 MV_t^b
 = \sum_{s=t+1}^{T^b}P(t,s)K^bN^b + P(t,T_b)N^b
\]
where $P(t,s)$ is the value of the zero-coupon bond from the mean-field Libor market ESG. The book value is determined by the strict lower-of-cost-or-market (LCM) principle, that is 
\[
 BV_t^b = \min\Big(BV_{t-1}^b, MV_t^b\Big).
\]
When a bond $b$ with nominal $N^b$  and maturity $T^b$ is bought at time $t$, the initial book value is $BV_t^b = N^b$, and the coupon factor $K^b$ follows from the requirement $MV_t^b = N^b$ and the prevailing yield curve at $t$ up to $T^b$. I.e., bonds are bought \emph{at par}. 

An equity position, $e$, consists of a market value, a book value, a (constant) volatility factor and a (constant) dividend factor. The latter is relevant for the company's surplus which is calculated according to local generally accepted accounting principles (local GAAP), since the dividend affects the book value return. The market value development is given by the geometric Brownian motion \eqref{e:GBM}. The book value is given by the strict LCM principle, that is $BV_t^e = \min(BV_{t-1}^e, MV_t^e)$.

Properties are modelled similar to equities with two distinctions: The dividend factor is interpreted not as a dividend but as rental income. Second, properties, $p$, are split into building and land value, $BV^p = BV^{bu}+BV^{lv}$, and the building value has a depreciation time, $T^p$, such that $BV^{bu}_{T^p}=0$.  
This depreciation time (which is usually not more than $30$  years) has to be provided as part of the initial data. The depreciation is linear, i.e.\ 
according to $(1-(s-t+1)(T^p-t+1))BV^{bu}_{t-1}$ for $t\le s\le T^p$,
and the strict LCM applies. Hence the book value development is given by 
\begin{align*}
 BV_t^p 
 &= BV_t^{bu}+BV_t^{lv} \\
 BV_t^{bu}
 &= \min\Big(
 (1-\frac{1}{T^p-t+1})BV^{bu}_{t-1}  ,  
  MV_t^{bu}\Big)
  \\
 BV_t^{lv}
 &= \min\Big(
  BV^{lv}_{t-1}  ,  
  MV_t^{lv}\Big)
\end{align*}
where $MV_t^{bu}$ and $MV_t^{lv}$ follow \eqref{e:GBM} along the \emph{same} Brownian path. 
Consequently properties often carry comparatively large amounts of unrealized gains $UG^p_t = MV_t^p-BV_t^p$.

\subsection{Liability module}\label{s:LM}
Let $L_t$ \label{ref: L_t} denote the book value of liabilities at time $t$ and assume that $L_t = 
LP_t + SF_t$ is a sum of two items:
Firstly, the life assurance provision, $LP_t = V_t + DB_t^{0} + DB_t$\label{ref: LP_t}; here $V_t$ is the mathematical reserve at time $t$ which depends only on the survival rates of policyholders but not on future surplus  declarations; the term $DB_t^{0}$ contains discretionary benefits that have been credited before valuation time $t=0$ and depends only on the survival rates of policyholders but not on future surplus  declarations; and the term $DB_t$  contains discretionary benefits that have been credited at times $1\le s <t$ and have not yet been paid out. 
Secondly,  the surplus fund, $SF_t$\label{ref: SF_t}, which consists of those profits that have not yet been declared to policyholders, cf.\ diagram~\eqref{d:gs} for a schematic overview.  

This set-up follows the same logic as \cite{Gerstner08} where $V_t$, $DB_t^{0}+DB_t$, and $SF_t$ are referred to as the actuarial reserve,  allocated bonus, and free reserve (buffer account), respectively.  Notice that we have split the allocated bonuses according to $DB_t^{0}$ and $DB_t$.

Since all cash-flows lead to a corresponding increase or decrease of the cash position in the asset portfolio it follows that   $L_t$ has to coincide at all times with the total book value, $BV_t$, of assets under management  
\begin{equation} \label{e:BV = LP + SF}
    \sum_{a\in\mathcal{A}_t} BV_t^a
    =
    BV_t 
    = L_t 
    = LP_t+SF_t
\end{equation}
and the verification of this equality is implemented as an automated test along all scenarios and for all time steps (cf.~\cite[A.~2.2]{HG19}). 
 
Notice that the equity position in the balance sheet model of \cite{Gerstner08} is a hybrid of free capital, given in the present notation as the difference  $BV_t - L_t$, and hidden reserves, $UG_t = MV_t-BV_t$. We do retain $UG_t$ in the projection since this is indispensable for the formulation appropriate management rules for best estimate calculation.

The ALM model assumes that the following quantities are given as deterministic functions of time:
\begin{itemize}
\item 
premium payments: $pr_t$
\item 
guaranteed benefits, including surrender and paid-up payments, due to $V_t$ and $DB_t^{0}$: $gbf_t$
\item 
expense payments: $exp_t$
\item 
mathematical reserve: $V_t$
\item 
previously allocated bonuses: $DB_t^{0}$
\item 
technical interest rate: $\rho_t$.
\end{itemize}
In fact, all these quantities are given on the level of model points such that the quoted values are the aggregate sums. Moreover, these quantities have been constructed according to classical actuarial assumptions such that benefits, premiums and technical reserves are consistent. 

The liability module gives rise to two management rules concerning the quantities $\nu_t$ and $\eta_t$ in \eqref{e:cred}. These rules are specified in Section~\ref{sec:MR}.

\subsection{Management rules and profit declaration}\label{sec:MR}
The company's management has a certain freedom to act based on market information and the state of assets and liabilities. This freedom concerns decisions about strategic asset allocation and profit declaration. A general understanding of the mechanics of profit declaration may be gained from Figure~\ref{d:gs} where the connection between gross surplus, management actions, profit declarations and cash-flows is summarized in a diagram. 

We recall the definition of the gross surplus \eqref{e:gs_g}, and note that 
\[
 ROA_t = \sum_{a\in\mathcal{A}_{t-1}}\Big(cf_t^a + \Delta BV_t^a\Big) + F_{t-1} C_{t-1} 
\]
where $\mathcal{A}_{t-1}$ denotes assets under management at $t-1$ other than cash, $cf_t^a$ is the asset's cash flow at $t$, i.e.\ coupon, nominal, dividend or rental income,  $\Delta BV_t^a = BV_t^a - BV_{t-1}^a$ is the change in book value according to the strict LCM principle, and  $C_{t-1}$ denotes the amount of cash at $t-1$.

\subsubsection{Strategic asset allocation}

\begin{mr}
Bonds are valued according to the augmented lower of cost or market principle, i.e.\ the initial book value is carried forward until the bond is sold or reaches its maturity. Equity and property  positions are valued according to the strict lower of cost or market principle, $BV_t^{\cdot} = \min(BV_{t-1}^{\cdot}, MV_t^{\cdot})$. 
\end{mr}

Notice that this rule implies that only bonds can have unrealized losses, all other assets' unrealized gains must be non-negative. 

\begin{mr}
New bonds are bought at par with a time to maturity of $10$ years. 
\end{mr}

\begin{mr}\label{MR:SAA}
The strategic asset allocation is kept approximately constant: the market value ratios (that is, cash amount over total market value, bonds over total market value, equities over total market value, and properties over total market value) are kept constant up to a predefined deviation. When an asset class breaches this bound, the portfolio is rebalanced such that the original targets, given by the market value ratios at $t=0$, are restored. The rebalancing is such that placement of assets with minimal unrealized gains is prioritized (to avoid unintended book value return). 
\end{mr}

\subsubsection{Negative surplus}
\begin{mr}\label{MR:UG}
When the gross surplus is negative, unrealized gains are realized until the gross surplus equals $0$, or no more positive unrealized gains exist. The selling order is: bonds before equity before property, and within those classes positions with large amounts of positive unrealized gains are prioritized. 
\end{mr}  

This rule is in place in order to avoid shareholder capital injections as much as possible. The gross surplus is updated after the application of this rule. 



\subsubsection{Positive surplus}
The following rules apply if the gross surplus that is calculated by the model after the portfolio has been aligned according to the strategic asset allocation is positive. The gross surplus is updated after the application of each of the following rules. 

Let $\vartheta = SF_0/LP_0$. 
Let $\tau_{t}$ denote the declared total participation rate at $t$. Notice that a participation rate of $\tau_t$ means that the amount of bonus declaration is given by 
\[
 ph_t^* 
 = \nu_t \cdot gph\cdot gs_t^+ + \eta_t\cdot SF_{t-1}
 = \sum_{x\in\mathcal{X}_{t}}\Big(\tau_t-\rho^x\Big)_+V_{t-1}^x
\]
where $\mathcal{X}_{t}$ is the set of model points $x$ which are active  at $t$, and $\rho^x$ and $V_{t-1}^x$ are the technical interest rate and previous mathematical reserve.  

Let $v = 1/100$. We define the target rate   
\[
 \tau_t^* = \min\Big( (\tau_{t-1}+L_t^{10})/2, \tau_{t-1}+v\Big)
\]
and the target amount of participation 
$ta_t = \sum_{x\in\mathcal{X}_{t}}(\tau_t^*-\rho^x)_+ V_{t-1}^x$.
This choice for $v$ means that the target participation rate increases at most $1$ percentage point from the previous participation rate thus avoiding large upwards jumps. 
The target  is defined in terms of the previous participation rate, $\tau_{t-1}$, and the prevailing $10$Y-forward rate, $L_t^{10}$, it is therefore a combination of previous profit participation and general market expectation.

\begin{mr}\label{MR:UG_ctrl}
If $\tau_t > \tau_t^*$ and there are bonds $b\in\mathcal{A}_t$ with unrealized losses, $UG_t^b<0$, then these losses are incurred until either $\tau_t = \tau_t^*$ or all unrealized gains are non-negative. 
\end{mr}

\begin{mr}\label{MR:nu-eta} 
Assume $gs_t>0$. We distinguish two cases: 
\begin{enumerate}
\item
If $gph\cdot gs_t^+\ge ta_t$:
\begin{align*}
 \nu_t 
 &= \min\Big(
    1 , 
    \sum_{x\in\mathcal{X}_t}(\tau_t^* - \rho^x)_+V_{t-1}^x / (gph\cdot gs_t^+)
 \Big) \\ 
 \eta_t^{(1)} &= 0
\end{align*}
\item 
If $gph\cdot gs_t^+ < ta_t$:
%
\begin{align*}
    \nu_t &= 1\\
    \eta_t^{(0)} 
    &=
    \Big( 
        \sum_{x\in\mathcal{X}_{t}}(\tau_t^* - \rho^x)_+ V_{t-1}^x - gph\cdot gs_t^+
    \Big)_+\Big/SF_{t-1}\\
    \eta_t^{(1)}
    &=\min\Big(\frac{1}{2}, \eta_t^{(0)}\Big) 
\end{align*} 
unless $SF_{t-1}=0$, in which case we set $\eta_t^{(1)}=0$.
\end{enumerate}
Given $\nu_t$ according to (1) or (2), we set
\begin{align*}
 \eta_t^{(2)} 
 &= 
 \frac{(1-\nu_t-\theta\nu_t)\cdot gph \cdot gs_t^+
          + SF_{t-1}
          - \theta(V_t+DB_t^{0} + DB_{t-1} - ph_t - sg_t^*)}{(1+\theta)SF_{t-1}}
\end{align*} 
unless $SF_{t-1}=0$, in which case we set $\eta_t^{(2)}=0$.
Finally, we define 
\[
 \eta_t 
 =  
 \min\Big(
  \max(\eta_t^{(1)}, \eta_t^{(2)}) , 1\Big).
\]
\end{mr} 

The purpose of this rule is to follow the above defined target rate by managing the surplus fund allocations, however to reach this target not more than half of the existing surplus fund is provided. 
The exception to this limit on surplus fund injections is the case $\eta_t^{(2)} > \eta_t^{(1)}$. This last term is to ensure that the inequality 
\begin{equation}
\label{e:theta-con}
 SF_{t-1} + gph\cdot gs_t^+ - ph_t^*
 = SF_t
 \le \theta LP_t 
 = \theta\Big(
   V_t+DB_t^{0} + DB_{t-1} + ph_t^* - ph_t - sg_t^*
   \Big)
\end{equation}
holds for all time.

\subsubsection{Order of management rules}

\begin{mr}
The rules are applied in the order in which they are stated. 
\end{mr} 

As a consequence   it may happen, within one accounting step, that, e.g., Rule~\ref{MR:UG} is carried out after the rebalancing step in Rule~\ref{MR:SAA}. Thus capital gains may be realized by, e.g., selling an equity position so that the distribution of assets may no longer be in line with the strategic asset allocation. In such a case a misalignment of asset positions is carried forward along one accounting year and then rebalanced at the end of this year. This rule is chosen nevertheless in this form since short term misalignment is acceptable.

\section{Phenomenological assumptions and numerical evidence}
\label{sec:assump}

This section is concerned with phenomenological assumptions based on numerical experience. The underlying numerical experiments are carried out with respect to realistic life insurance data (cf.\ Section~\ref{sec:numALM}). This implies that the assumptions which are formulated below may hold only in an approximate sense. We consider a given statement or assumption as fulfilled if it holds up $0.5\,\%$ of the initial market value, $MV_0$, of the given asset portfolio. 

In the following the assumptions are compared to the numerical model over a range of nine parameters. We consider the base case parameterization from Section~\ref{sec:numALM}, and then vary the initial amount of unrealized gains and the amount of guaranteed premium payments. In the base case we have $UG_0/BV_0 = 0.05$ and premium payments are scaled by the unit factor $\pi_0 = 1$. In order to test the assumptions over a wide range of parameters we consider all combinations of 
\begin{equation}
\label{e:UGstart}
 UG_0/BV_0:
 \quad -0.10, \quad 0.05, \quad 0.20, 
\end{equation}
and 
\begin{equation}
\label{e:pi}
 \pi_0:
 \quad 0.95, \quad 1, \quad 1.05,
\end{equation}
which lead to very different relationships between asset and liability portfolios.   Indeed, the factor $UG_0/BV_0$ controls the market value of assets available to cover the liabilities. Owing to equation \eqref{e:gamma_def}, varying the premium income is equivalent to varying the average technical interest rate. Thus the variation given by $\pi_0$ can be interpreted as a variation of the relationship between the prevailing yield curve and the guaranteed minimum rate, i.e.\ as a variation of the relative conditions between current economy and liabilities.  

\subsection{Assumptions, evidence and discussion}

\begin{ass}
    \label{ass:runoff}
    The liabilities are in run-off such that:
    \begin{enumerate}
    \item 
    The projection horizon $T$\label{ref: T} corresponds to the run-off time of the liability portfolio, that is  
    $SF_T = LP_T = UG_T = 0$.
    \item 
    The expected life assurance provisions $E[LP_t]$ decrease geometrically: there is a fixed $1\le h < T$\label{ref: h} such that $E[LP_t]$ can be approximated by $\widehat{LP}_t := l_t^h\, LP_0$ where $l_t^h := 2^{-t/h}$ for $t<T$ and $l_T^h := 0$. 
    \item 
    Moreover, $h$ can be approximated as the duration of liability cash-flows: 
    \begin{align} 
 h_0 &= 
 \frac{\sum_{t=1}^T \,t\, P(0,t)\,(gbf_t+exp_t-pr_t)}{\sum_{t=1}^T P(0,t)\,(gbf_t+exp_t-pr_t)} \notag \\  
 cf_t^{BS}
 &= \phi(t, 2h_0, h_0/2)\max(FDB_0, 0)
 \notag \\
 h &= 
  \frac{\sum_{t=1}^T \,t\, P(0,t)\,(gbf_t+exp_t-pr_t + cf_t^{BS})}{\sum_{t=1}^T P(0,t)\,(gbf_t+exp_t-pr_t + cf_t^{BS})} 
  \label{e:h_liab}
    \end{align}
    \end{enumerate}
\end{ass}


Here $\phi(\cdot,2h_0,h_0/2)$ denotes the normal density with mean $2h_0$ and standard deviation $h_0/2$. The idea of $cf_t^{BS}$ is to estimate the effect of future discretionary benefits on the duration of the liabilities. The simplest, albeit quite crude, estimator for $FDB$ is $FDB_0$. Hence these cash-flows are expected to be approximately proportional to $FDB_0$, and we consider only the positive part since future discretionary benefits are positive. The density function serves to distribute the effect of $FDB_0$ over time. The duration, $h_0$, of guaranteed cash-flows is now taken as an estimator for the peak of these cash-flows: the reasoning is that $h_0$ is already an estimator for the time it takes for the declared profits to accumulate, and $2h_0$ is the point where the policyholder payments from these declared benefits peak. 

The first part of assumption~\ref{ass:runoff} is trivial once $T$ is chosen sufficiently large. However, that the run-off is geometric and is, moreover, determined (approximately) by formula~\eqref{e:h_liab} requires some evidence. This is provided by Figure~\ref{fig:run-off}. Clearly, the above arguments for $h_0$, $cf_t^{BS}$ and $h$ are very heuristic, nevertheless the empirical evidence shows that the overall effects are as desired and the factor \eqref{e:h_liab} is a reasonable parameter controlling the geometric run-off. 

\begin{figure}[ht]
\includegraphics[width=5.2cm]{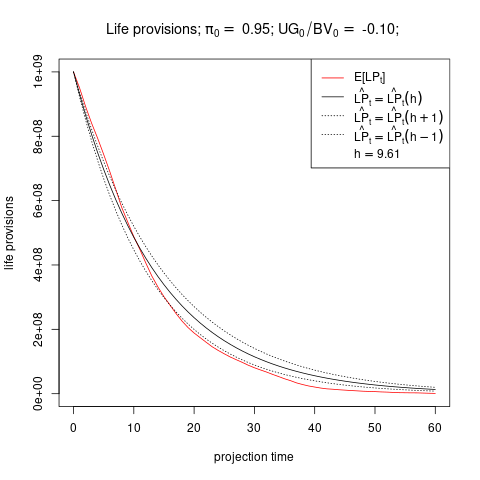}
\includegraphics[width=5.2cm]{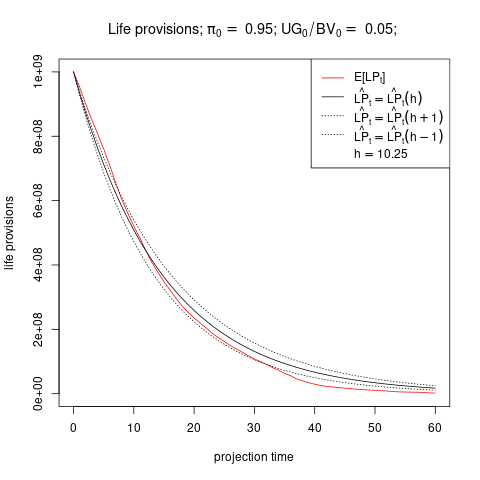}
\includegraphics[width=5.2cm]{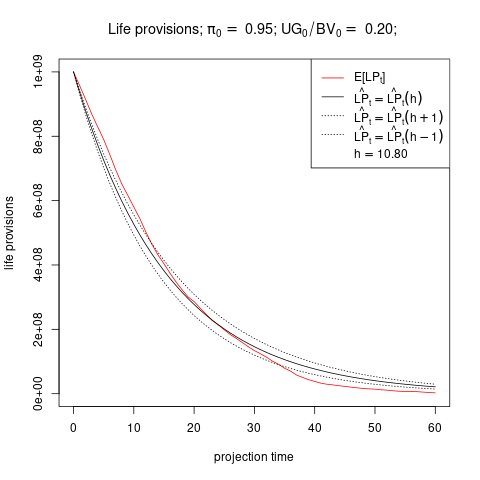}

\includegraphics[width=5.2cm]{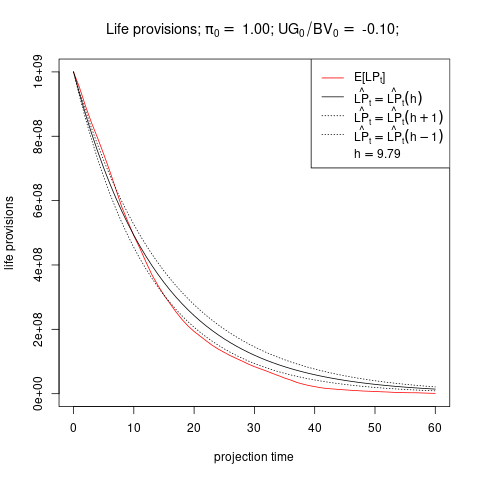}
\includegraphics[width=5.2cm]{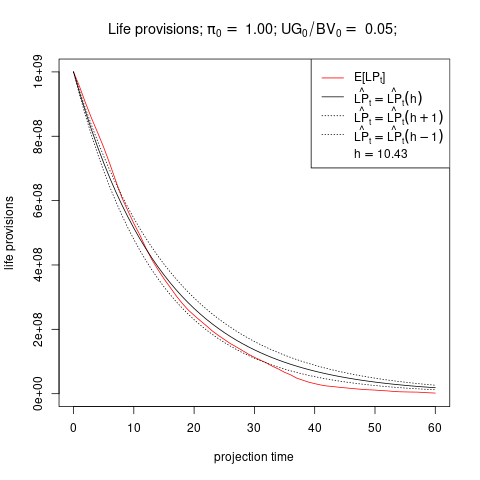}
\includegraphics[width=5.2cm]{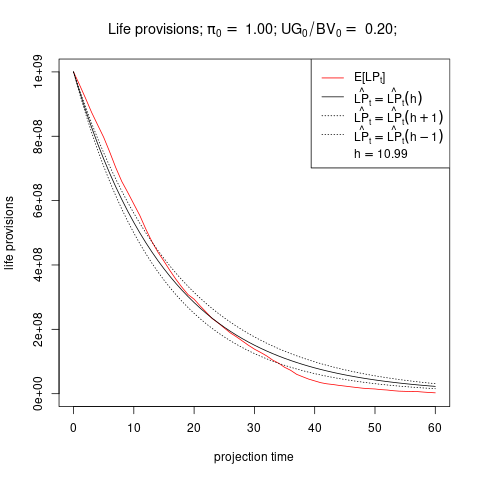}

\includegraphics[width=5.2cm]{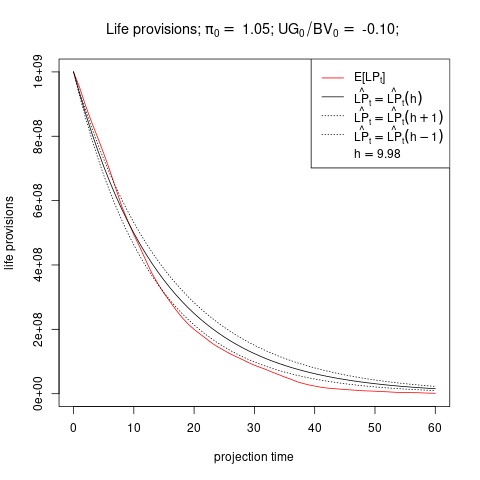}
\includegraphics[width=5.2cm]{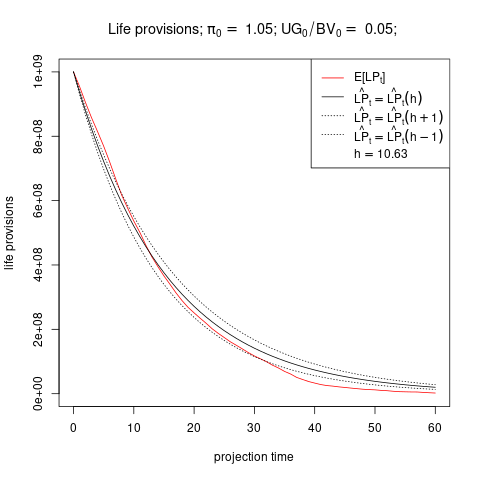}
\includegraphics[width=5.2cm]{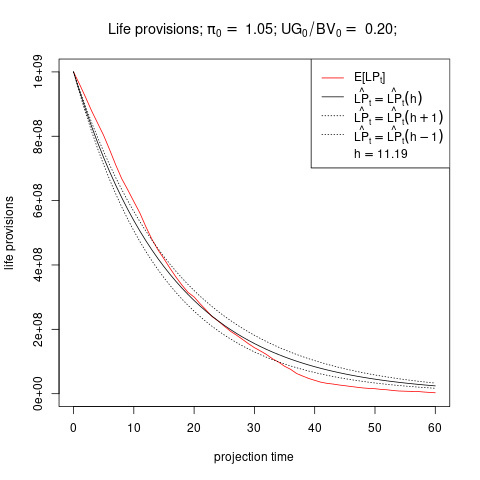} 
\caption{Numerical evidence for Assumption~\ref{ass:runoff}. The red line, $E[LP_t]$, is calculated empirically by means of the model described in Section~\ref{sec:numALM}. The factor $h$ is calculated using known quantities at time $0$ according to \eqref{e:h_liab}, and the dotted lines are included to show the effect of a variation $h\pm1$.}
\label{fig:run-off}
\end{figure}

\begin{ass}
\label{ass2:sigma}
Let $h$ be given by \eqref{e:h_liab}.
Let $\sigma_0 = DB_0^{0}/LP_0$,
$\sigma_1 = \max(FDB_0/GB, 0.75\cdot\sigma_0)$,
$\sigma_t = \sigma_0 \max(h-t,0)/h + \sigma_1 \min(t/h, 1)$
and 
$\sigma^{DB}_t = \sigma_1 \min(t/h, 1)$. 
Then  $E[DB_t^{0}]$ and $E[DB_t]$ can be approximated by $\widehat{DB}_t^{0}$ and $\widehat{DB}_t$, respectively, where
\[ 
  \widehat{DB}_t = \sigma^{DB}_t\, \widehat{LP}_t.
\]
and  
\[ 
 \widehat{DB}_t^{0} + \widehat{DB}_t = \sigma_t \, \widehat{LP}_t
\]
for all $0\le t\le T$ .
\end{ass}
 
This assumption states that the initial fraction of declared benefits, $\sigma_0$, remains approximately constant in expectation, but may change over time linearly to $\sigma_1$ with a speed that is determined by $h$. Moreover, $E[DB_t]$ increases linearly from $0$ to $\sigma_1$ in $h$ time steps, whence $E[DB_t^{0}]$ cannot decrease too quickly. Numerical evidence for this assumption is provided in Figure~\ref{fig:ass2}.  This figure shows that the assumption appears to be satisfied only up to time $h$ and then holds rather poorly at later projection times. Nevertheless, we view this assumption as justified, since the overall behavior is represented correctly, and errors at later points in time are discounted.  Moreover, given that this assumption relies on the  calculation of $h$ in Assumption~\ref{ass:runoff} and the expectation future profits as estimated by the fraction $FDB_0/GB$ in the definition of $\sigma_1$ the results in Figure~\ref{fig:ass2} show a remarkable stability. 


\begin{figure}[ht]
\includegraphics[width=5.2cm]{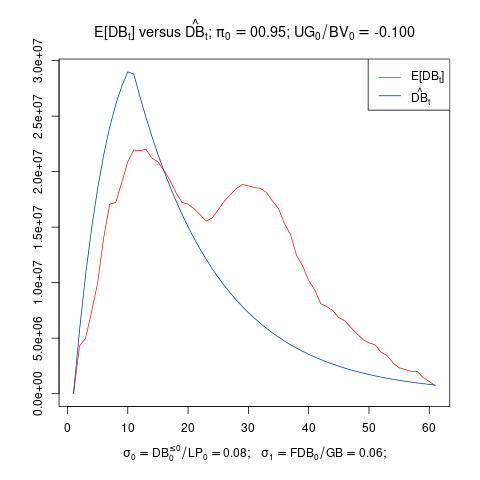}
\includegraphics[width=5.2cm]{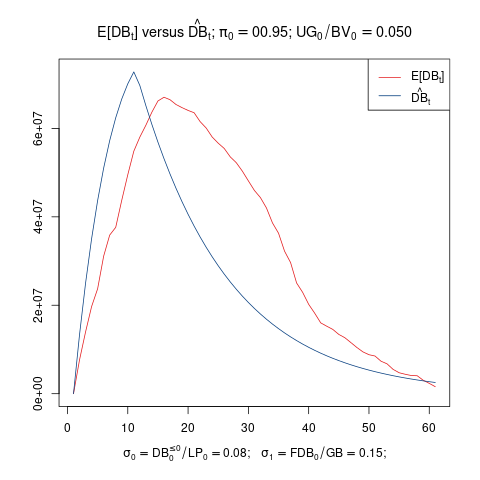}
\includegraphics[width=5.2cm]{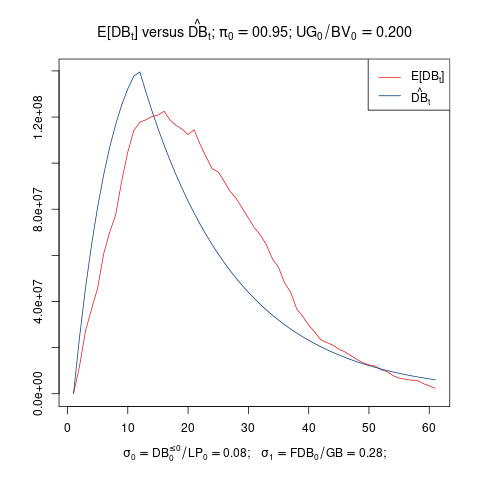}

\includegraphics[width=5.2cm]{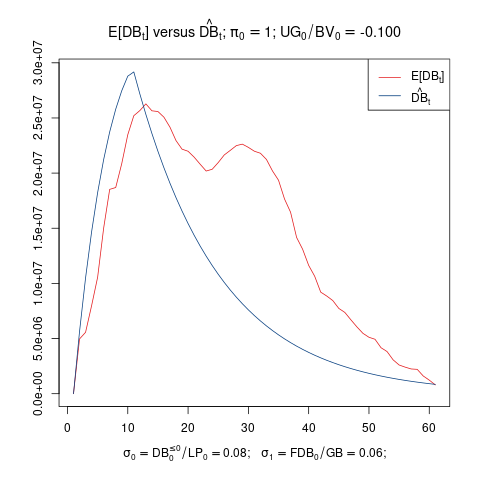}
\includegraphics[width=5.2cm]{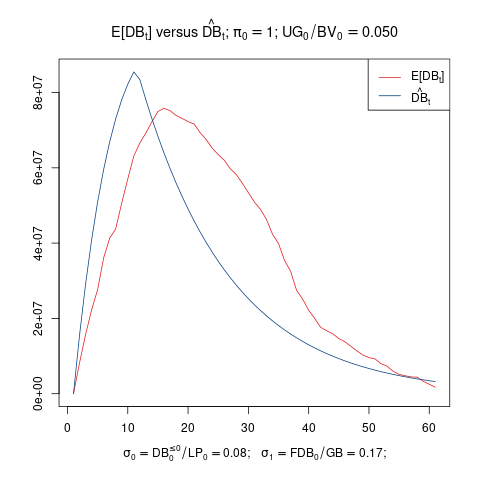}
\includegraphics[width=5.2cm]{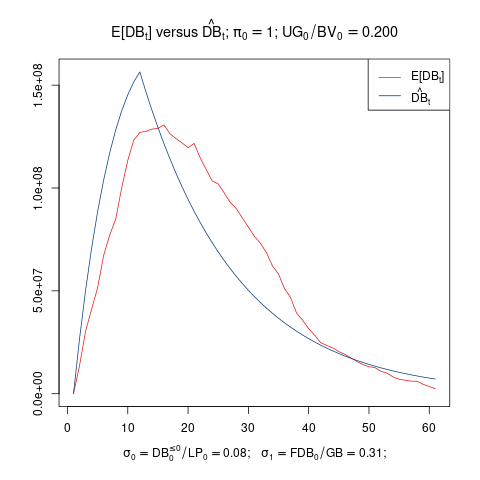}

\includegraphics[width=5.2cm]{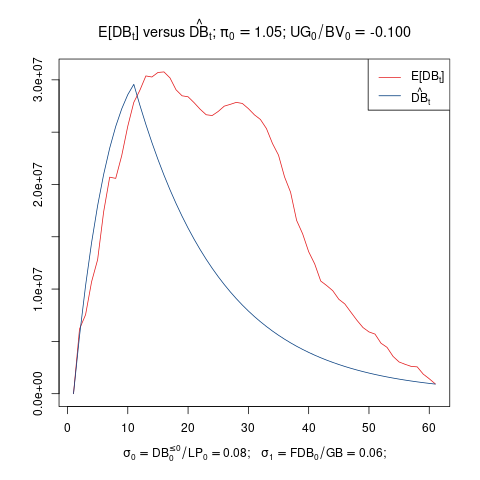}
\includegraphics[width=5.2cm]{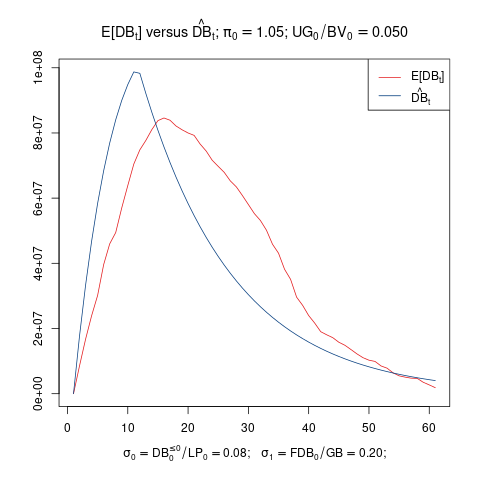}
\includegraphics[width=5.2cm]{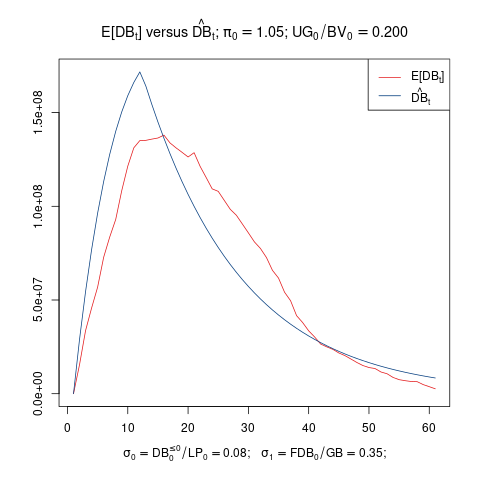}

\caption{Numerical evidence for Assumption~\ref{ass2:sigma}. $E[DB_t]$ is calculated numerically by means of the model described in Section~\ref{sec:numALM}. $\widehat{DB}_t$ is calculated with respect to $h$ as defined in \eqref{e:h_liab}.} 
\label{fig:ass2}
\end{figure}

\begin{ass} \label{ass3:SF} 
The relation $SF_0/LP_0 =: \vartheta$\label{ref: theta} remains constant in expectation: $E[SF_t] = \vartheta E[LP_t]$ for all $0\le t\le T$.  (Cf.~\cite[A.~3.10]{HG19})
\end{ass}

Assumptions~\ref{ass2:sigma} and \ref{ass3:SF} are  statements about time-homogeneity. Management rules concerning bonus declarations should remain reasonably constant in the long run such that $\sigma$ and $\vartheta$ do not vary too strongly. The relevant point in this context is that these quantities should not vary arbitrarily but follow from target rates set by management rules. Assumption~\ref{ass3:SF} is comparable to the assumption concerning the `annual interest rate' in \cite[Section~4.2]{Gerstner08}. Figure~\ref{fig:ass3SF} contains empirical observation for Assumption~\ref{ass3:SF}. Clearly, this figure only shows evidence for $E[SF]_t/E[LP_t]\le \vartheta $. Nevertheless, Assumption~\ref{ass3:SF} is kept in its present form since a statement such as $0.75\cdot\vartheta \le E[SF]_t/E[LP_t]\le \vartheta$, while being more accurate, introduces unnecessary complexity. 
 
\begin{figure}[ht]
\includegraphics[width=5.2cm]{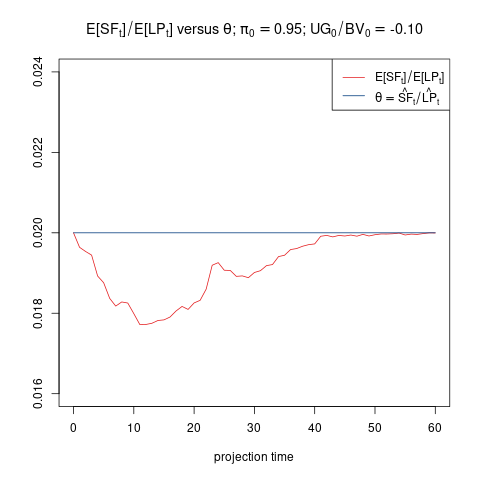}
\includegraphics[width=5.2cm]{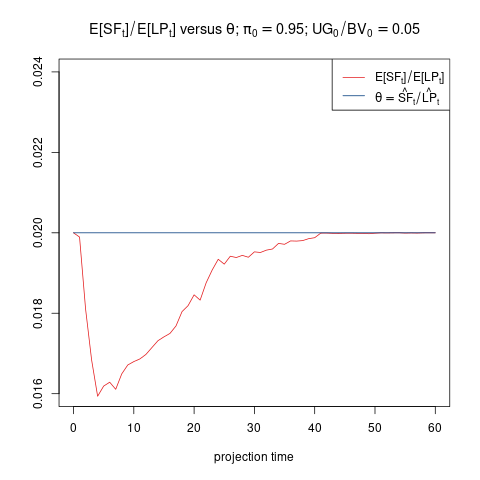}
\includegraphics[width=5.2cm]{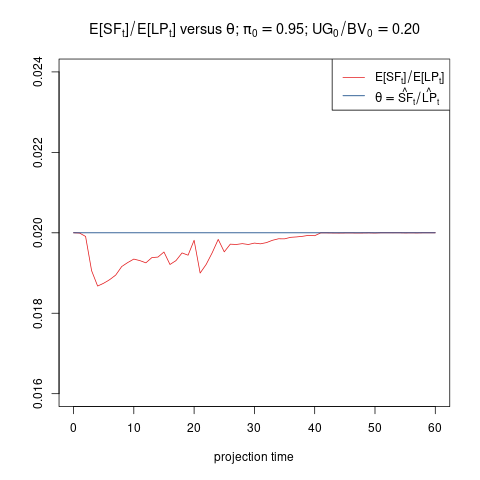}

\includegraphics[width=5.2cm]{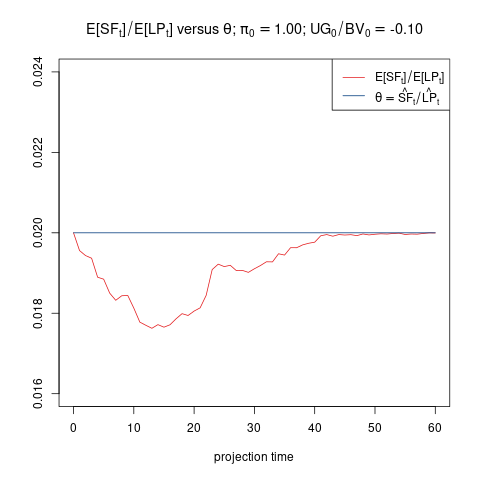}
\includegraphics[width=5.2cm]{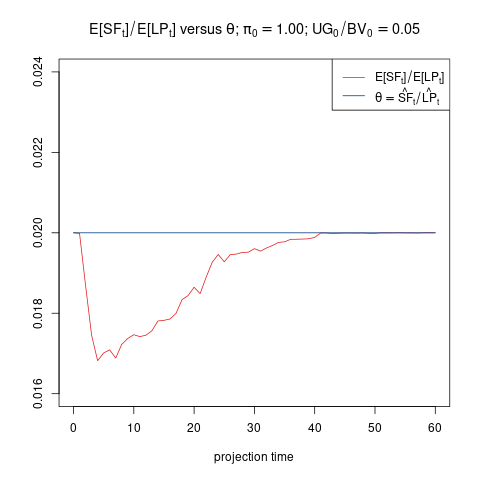}
\includegraphics[width=5.2cm]{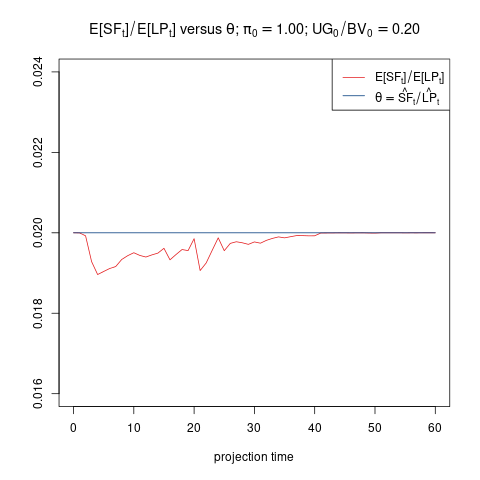}

\includegraphics[width=5.2cm]{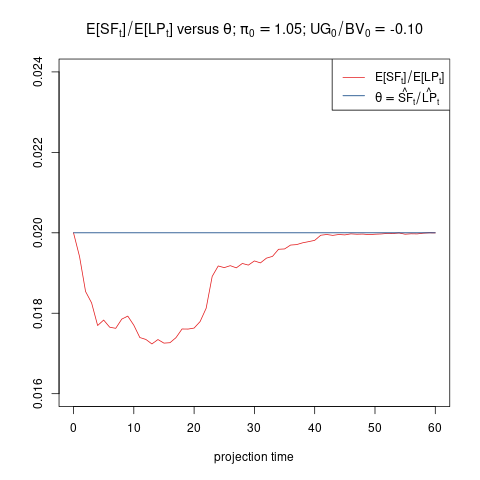}
\includegraphics[width=5.2cm]{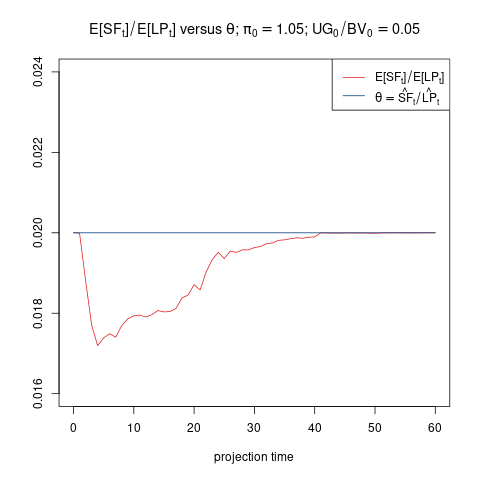}
\includegraphics[width=5.2cm]{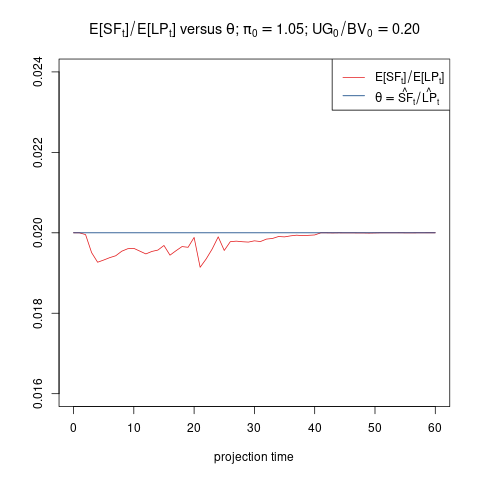}

\caption{Numerical evidence for Assumption~\ref{ass3:SF}. The quotient $E[SF_t]/E[LP_t]$ is calculated numerically by the model described in Section~\ref{sec:numALM}. The constant value is given by $SF_0/LP_0 = \vartheta = 2/100$.}
\label{fig:ass3SF}
\end{figure}


\begin{ass}
    \label{ass:surr}
    The surrender gains, $sg_t^* = \chi_t DB_{t-1}$, can be estimated on average with the same factor, $\gamma_t$, as the technical gains in \eqref{e:gs_g}: 
    $E[sg_t^*] \le \max(\gamma_t,0) E[DB_{t-1}] \le \max(\gamma_t,0) \widehat{DB}_t$.  
\end{ass}
The factor $\gamma_t$ contains mortality, cost and surrender margins as a fraction of the full life assurance provision, $LP_{t-1}$. It is therefore reasonable to expect that the same factor can be used as an upper bound on the surrender margin arising from declared bonuses, $DB_{t-1}$, alone. Since the technical gains may also be negative only the positive part, $\max(\gamma_t,0)$, is considered.



\begin{ass}\label{ass:cov}
\begin{enumerate}
\item 
The correlation between $B_t^{-1}$ and $DB_{t-1}$ is negative: $\operatorname{Cor}[B_t^{-1}, DB_{t-1}] \le 0$. 
\item 
The relative covariance of $B_t^{-1}F_{t-1}$ and $DB_{t-1}+SF_{t-1}$ satisfies 
\begin{align*}
 &\operatorname{Cov}\Big[ 
  \frac{B_t^{-1}F_{t-1}}{E\left[B_t^{-1}F_{t-1}\right]}, 
  \frac{DB_{t-1}+SF_{t-1}}{E\left[DB_{t-1}+SF_{t-1}\right]} 
 \Big] \\ 
 &\phantom{X} =
 \operatorname{Cor}\Big[B_t^{-1}F_{t-1}, DB_{t-1}+SF_{t-1} \Big]
 \cdot CV\Big[B_t^{-1}F_{t-1}\Big]\cdot CV\Big[DB_{t-1}+SF_{t-1}\Big] \\
  &\phantom{X} 
  \le 1
\end{align*}
for $1\le t\le T$ where  $CV$ denotes the  coefficient of variation. 
\end{enumerate}
\end{ass}

The first part of this assumption is straightforward  since a low discount factor is related to high nominal rates and this yields high nominal returns and consequently high reserves $DB_{t-1}$. 
In the long run this argument applies also to the second part of the assumption, i.e.\ to  the correlation of $B_t^{-1}F_{t-1}$ and $DB_{t-1}+SF_{t-1}$ since, after some time, the behaviour of $B_t^{-1}F_{t-1}$ should be dominated by $B_t^{-1}$. In the short run when the correlation might be positive, however, we expect the product, $CV[B_t^{-1}F_{t-1}]\cdot CV[DB_{t-1}+SF_{t-1}]$ of the coefficients of variation to be small. Numerical evidence for these heuristic ideas is provided in Figure~\ref{fig:ass_cov}. 

\begin{figure}[ht]
\includegraphics[width=5.2cm]{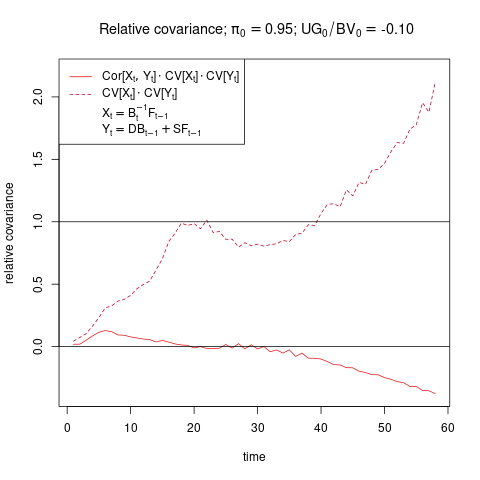}
\includegraphics[width=5.2cm]{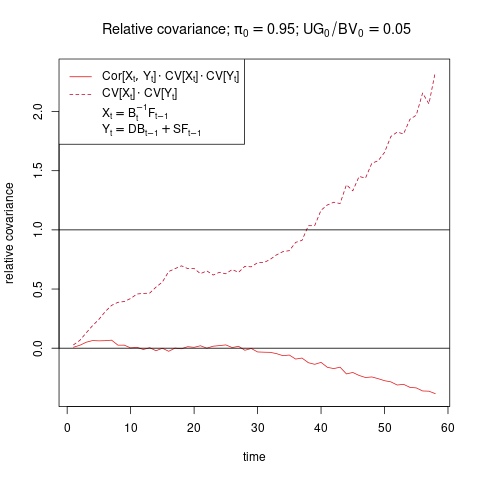}
\includegraphics[width=5.2cm]{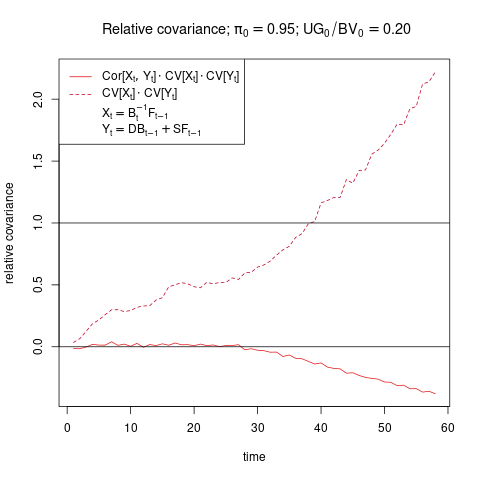}

\includegraphics[width=5.2cm]{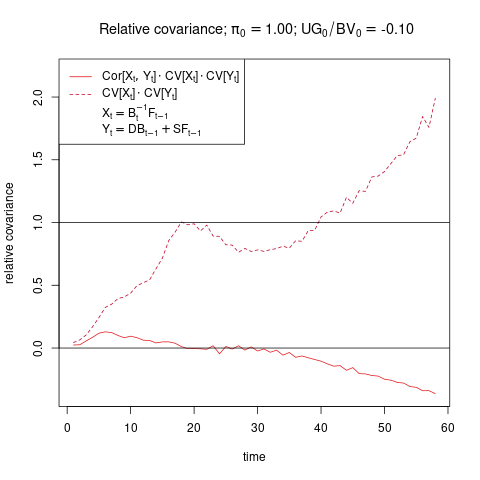}
\includegraphics[width=5.2cm]{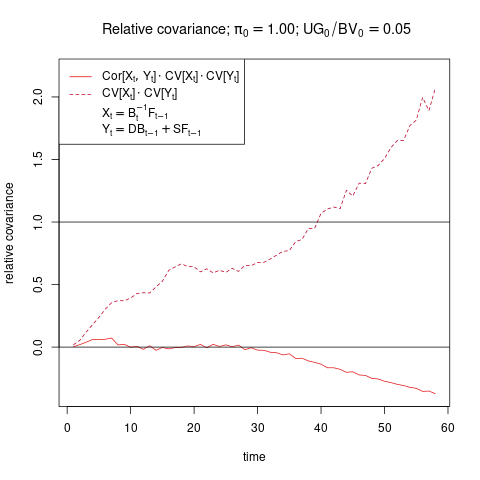}
\includegraphics[width=5.2cm]{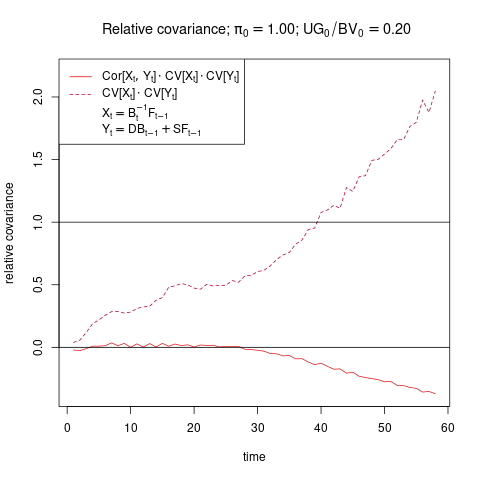}

\includegraphics[width=5.2cm]{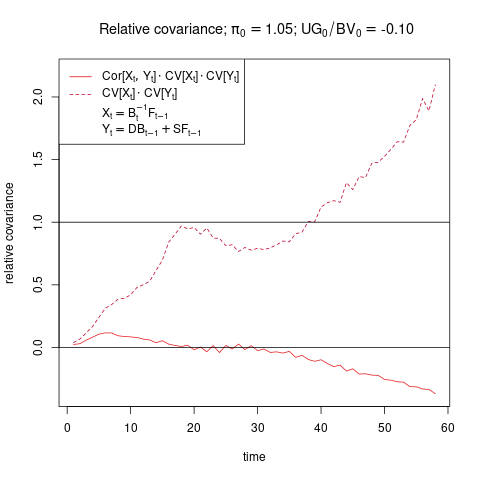}
\includegraphics[width=5.2cm]{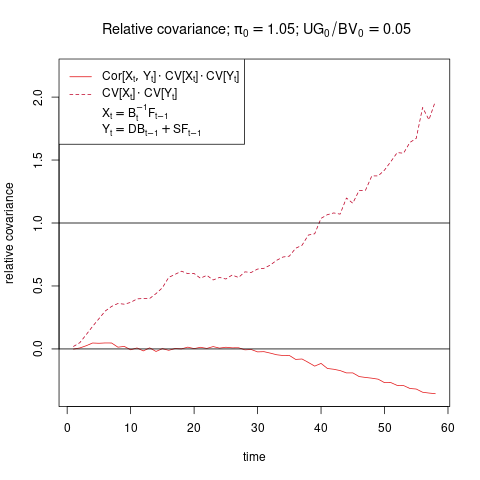}
\includegraphics[width=5.2cm]{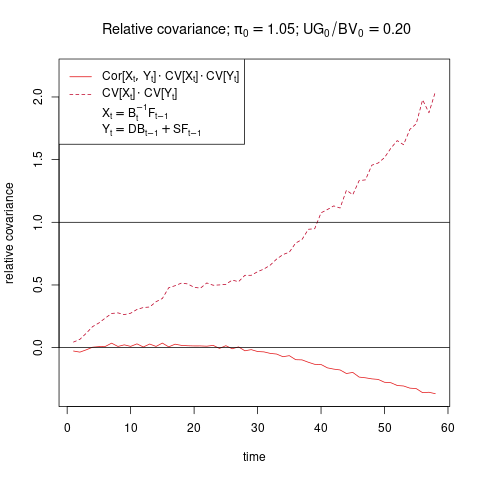}

\caption{Numerical evidence for Assumption~\ref{ass:cov}. The relative covariance  is calculated numerically by the model described in Section~\ref{sec:numALM}, additionally the product of the coefficients for variation is indicated by the dotted line.} 
\label{fig:ass_cov}
\end{figure}

According to ~\eqref{e:gs_g} the gross surplus is given by 
\begin{align*}
    gs_t
    =
    F_{t-1} (LP_{t-1} + SF_{t-1}) 
    + F_{t-1}UG_{t-1} - E[\Delta UG_t|\mathcal{F}_{t-1}] 
    + ROA_t - E[ROA_t|\mathcal{F}_{t-1}]
    - (\rho_t-\gamma_t) V_{t-1} 
\end{align*}
where $\rho_t$ and $\gamma_t$ may, in general, also depend on the stochastic interest rate curve via dynamic surrender. 
 
For the purpose of estimating terms $III$ and $COG$ in \eqref{e:fdb-rep}, we make the following simplifying assumptions. The principle idea behind these assumptions is that the main source of stochasticity in $gs_t$ is the forward rate $F_{t-1}$ whence all other quantities are replaced by their expected values.  The simplified model of $gs_t$ will be denoted by $\widehat{gs}_t$.

\begin{ass}
    \label{ass:tech}
    In $\widehat{gs}_t$ the technical interest rate $\rho_t$ and the technical gains $\gamma_t$ are deterministic functions of $t$. More precisely, $\rho_t$ is calculated as the $LPG_t$-weighted mean guaranteed rate, and $\gamma_t$ is calculated according to 
\begin{equation}
    \label{e:gamma_hat}
 \gamma_{t}(1-(\sigma_{t-1}-\sigma^{DB}_{t-1})) LPG_{t-1}
 = 
 \rho_{t}\Big(1-(\sigma_{t-1}-\sigma^{DB}_{t-1})\Big) LPG_{t-1} - \Delta LPG_t + pr_t - gbf_t - exp_t 
\end{equation} 
\end{ass}

Notice that \eqref{e:gamma_hat} coincides with \eqref{e:gamma_def} up to the substitution $V_{t-1} = (1-(\sigma_{t-1}-\sigma^{DB}_{t-1})) LPG_{t-1}$, which is consistent with Assumption~\ref{ass2:sigma}, and we omitted the surrender gains $sg_t^*$. This definition allows us to compute $\gamma_t$ a priori from the given (deterministic) data. Figure~\ref{fig:gamma} contains plots of $\gamma_t$ and also of $\gamma_t + E[sg_t^*]$, and confirms that the effect of $sg_t^*$ should be expected to be quite negligible.

\begin{figure}[ht]
\includegraphics[width=5.2cm]{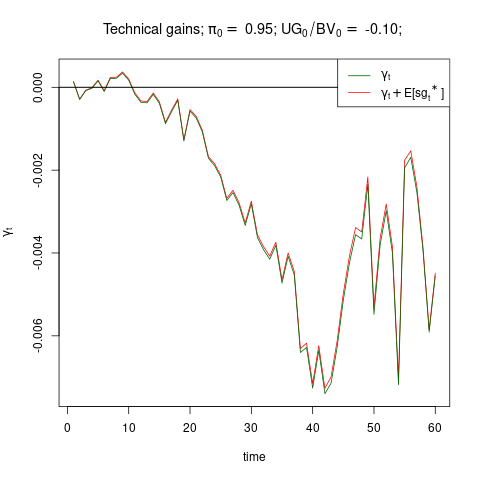}
\includegraphics[width=5.2cm]{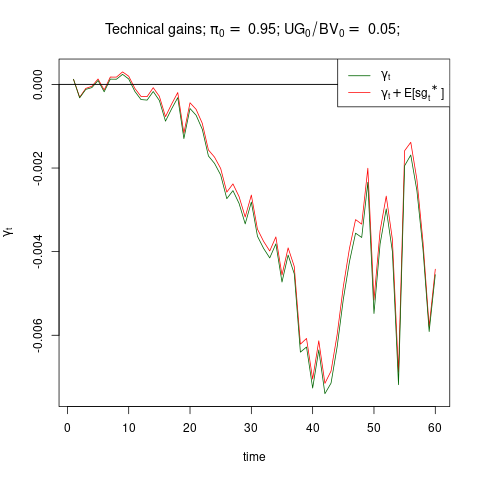}
\includegraphics[width=5.2cm]{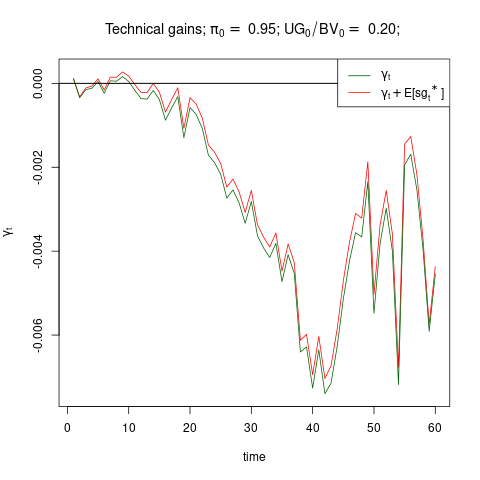}

\includegraphics[width=5.2cm]{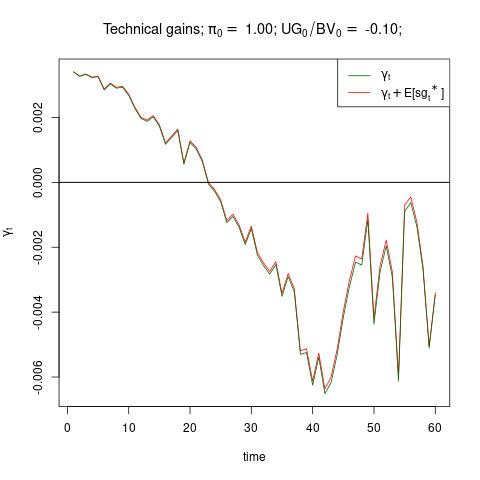}
\includegraphics[width=5.2cm]{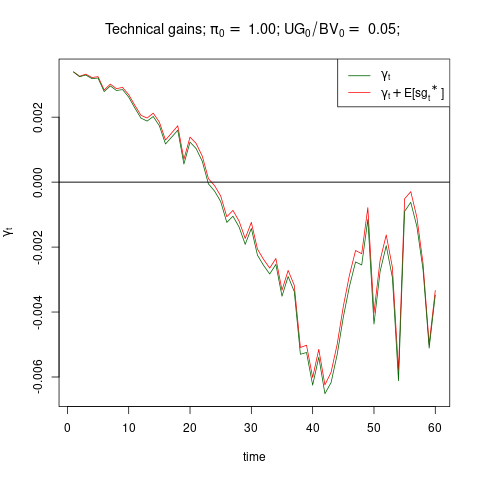}
\includegraphics[width=5.2cm]{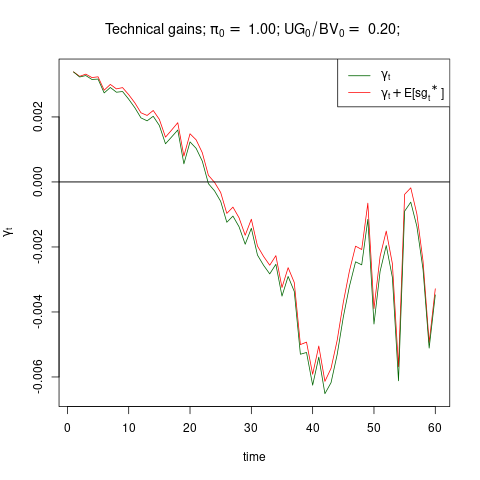}

\includegraphics[width=5.2cm]{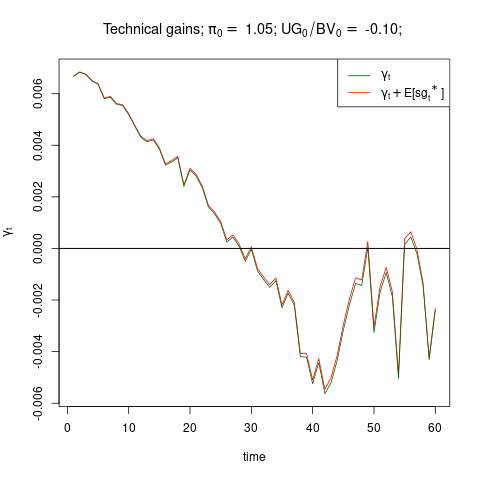}
\includegraphics[width=5.2cm]{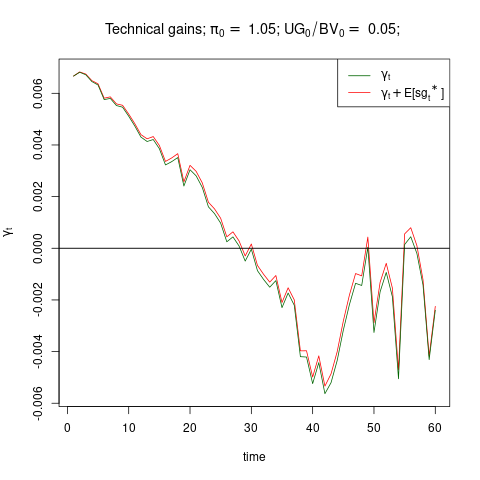}
\includegraphics[width=5.2cm]{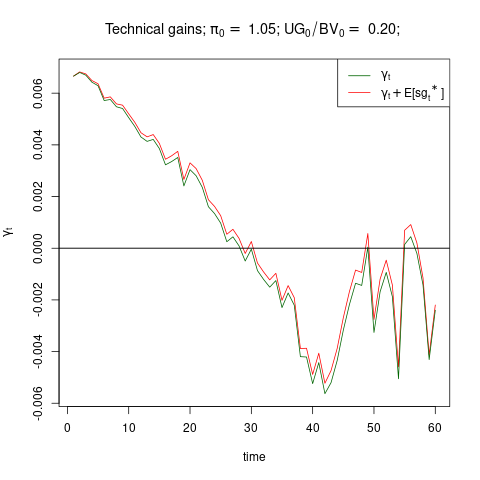}

\caption{Numerical evidence for Assumption~\ref{ass:tech}. It is shown that the effect of the model dependent quantity $sg_t^*$, which has been calculated by means of the numerical model described in Section~\ref{sec:numALM}, is negligible on average, such that \eqref{e:gamma_hat} is a reasonable approximation for the technical gains.}
\label{fig:gamma}
\end{figure}
 
\begin{ass}
\label{ass6:ROA}
In $\widehat{gs}_t$ the return $ROA_t$ is predictable, i.e.\ $\mathcal{F}_{t-1}$-measurable, and realizations of unrealized gains are determined by a fixed number $1< d < T$\label{ref: d}:
\begin{enumerate}[\up (1)]
    \item 
    $ROA_t - E[ROA_t|\mathcal{F}_{t-1}] = 0$;
    \item 
    $F_{t-1} UG_{t-1} - E[\Delta UG_t|\mathcal{F}_{t-1}] 
    = P(0,t)^{-1}(l_{t-1}^d - l_t^d) UG_0$
    where $l_t^d := 2^{-t/d}$ for $t<T$ and $l_T^d:=0$.
\end{enumerate}
Moreover, $d$ is given as the duration of the initial bond portfolio: $d = (\sum_{t=1}^T\sum_b t\,P(0,t)\,cf_t^b)/(\sum_{t=1}^T\sum_b P(0,t)\,cf_t^b)$ where $b$ runs over all initially held bonds. 
\end{ass}

To compare the first part of this assumption against numerical evidence, we consider the quantity  $(ROA_t - E[ROA_t|\mathcal{F}_{t-1}])/BV_{t-1}$ since $ROA_t$ is the total return on the book value $BV_{t-1}$. Figure~\ref{fig:ass6ROA1} shows observations of expected values and standard deviations, that is  $E[(ROA_t - E[ROA_t|\mathcal{F}_{t-1}])/BV_{t-1}]$ and  $SD[(ROA_t - E[ROA_t|\mathcal{F}_{t-1}])/BV_{t-1}]$. While the expected values vanish due to the tower property of the conditional expectation, the fact that also the standard deviations are (relatively) small is significant for Assumption~\ref{ass6:ROA}.

The second part of Assumption~\ref{ass6:ROA} is studied numerically in Figure~\ref{fig:ass6ROA2}. Again it can be observed that the assumption is not satisfied accurately, but nevertheless the overall behaviour is reflected quite well. 

\begin{figure}[ht]
\includegraphics[width=5.2cm]{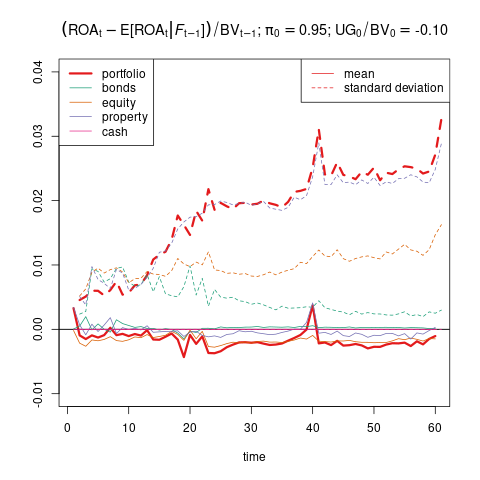}
\includegraphics[width=5.2cm]{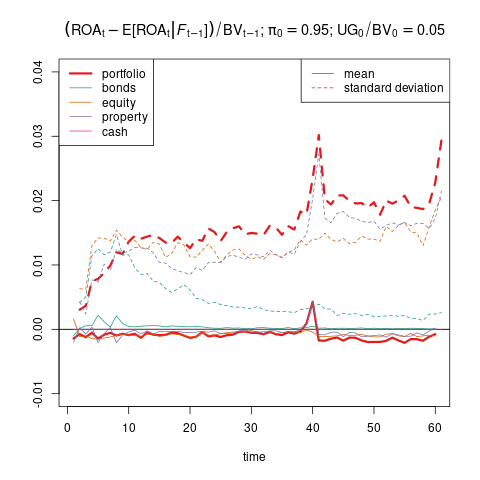}
\includegraphics[width=5.2cm]{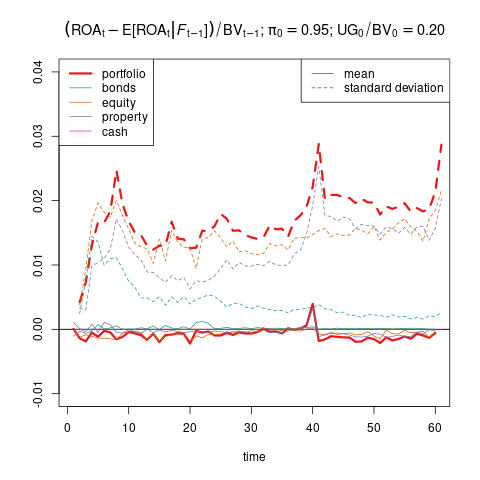}

\includegraphics[width=5.2cm]{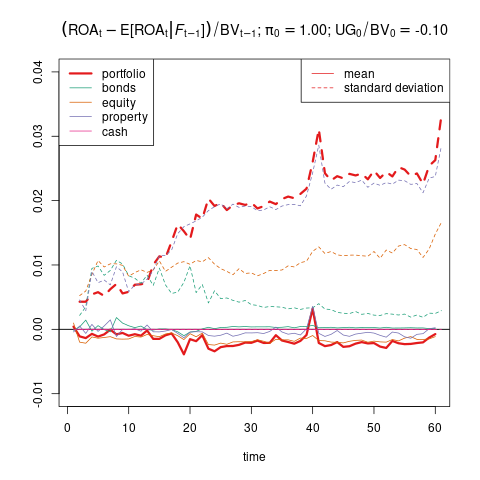}
\includegraphics[width=5.2cm]{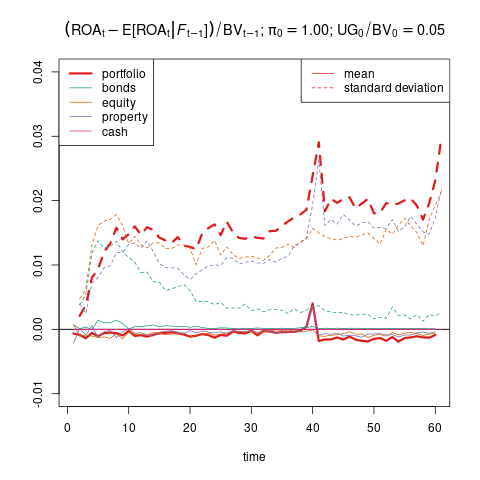}
\includegraphics[width=5.2cm]{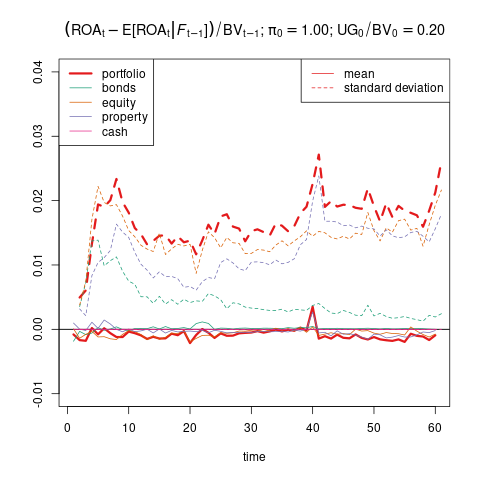}

\includegraphics[width=5.2cm]{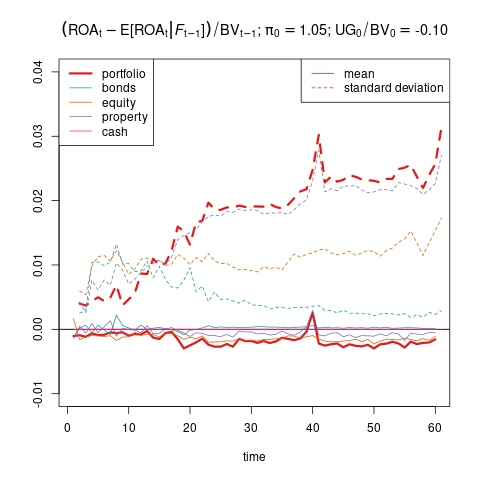}
\includegraphics[width=5.2cm]{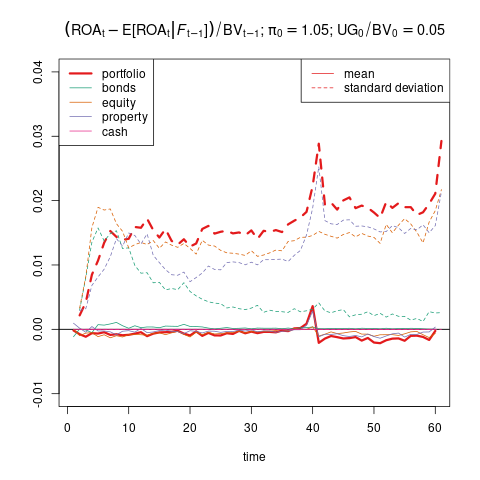}
\includegraphics[width=5.2cm]{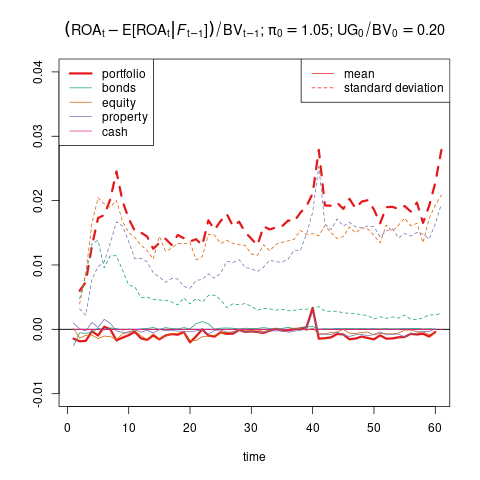}

\caption{Numerical evidence for the first part of Assumption~\ref{ass6:ROA}.} 
\label{fig:ass6ROA1}
\end{figure}

\begin{figure}[ht]
\includegraphics[width=5.2cm]{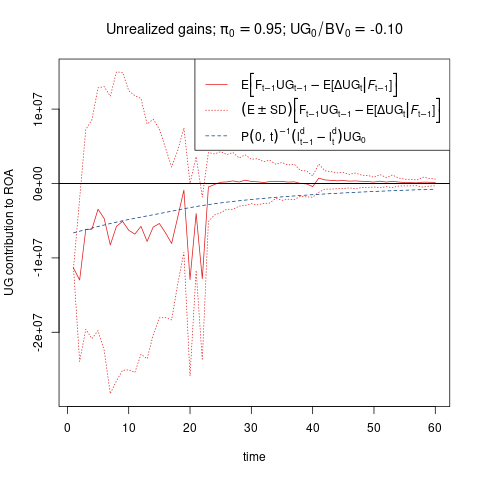}
\includegraphics[width=5.2cm]{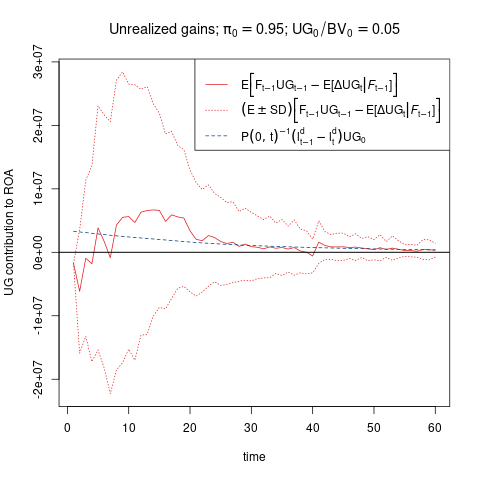}
\includegraphics[width=5.2cm]{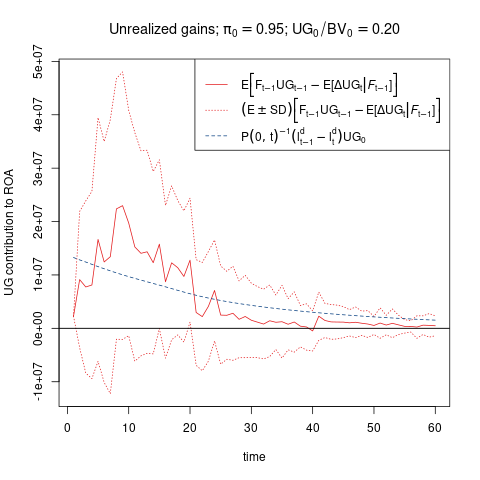}

\includegraphics[width=5.2cm]{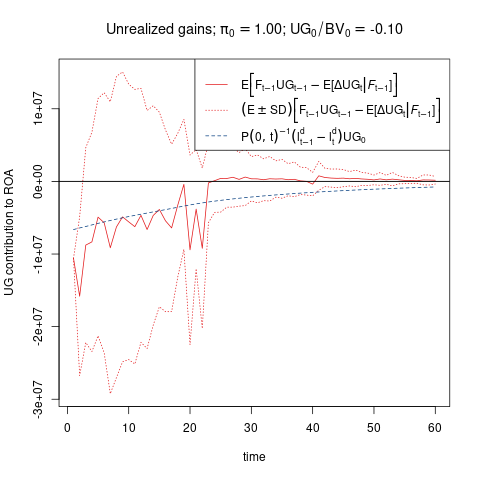}
\includegraphics[width=5.2cm]{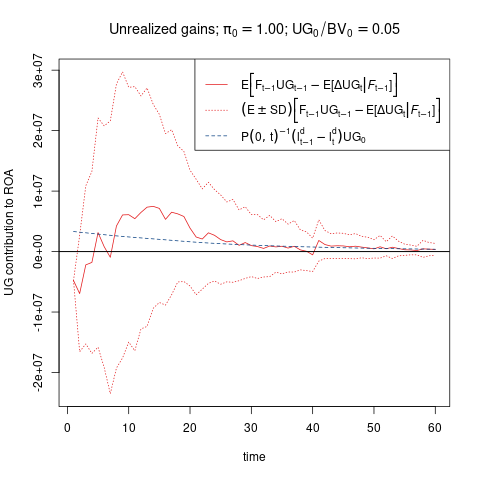}
\includegraphics[width=5.2cm]{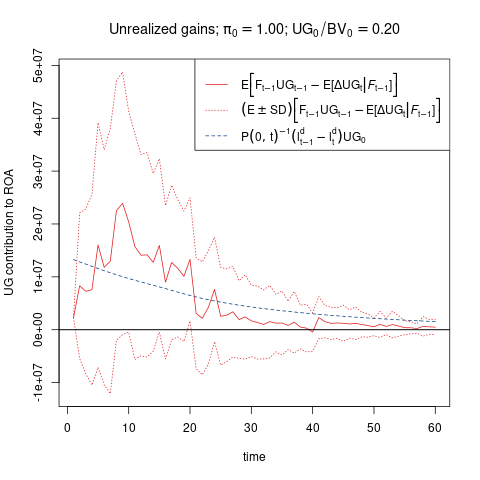}

\includegraphics[width=5.2cm]{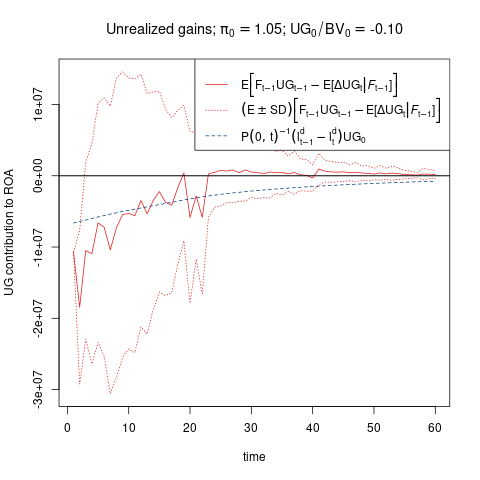}
\includegraphics[width=5.2cm]{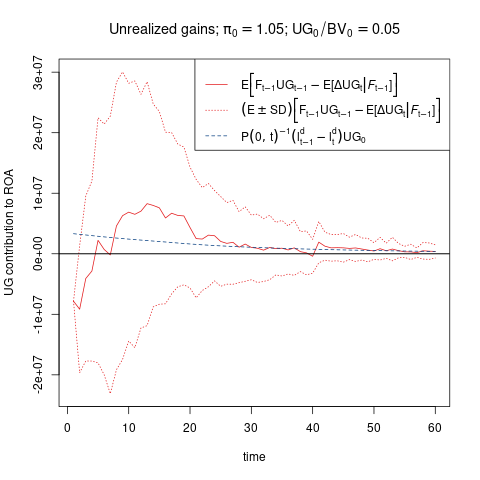}
\includegraphics[width=5.2cm]{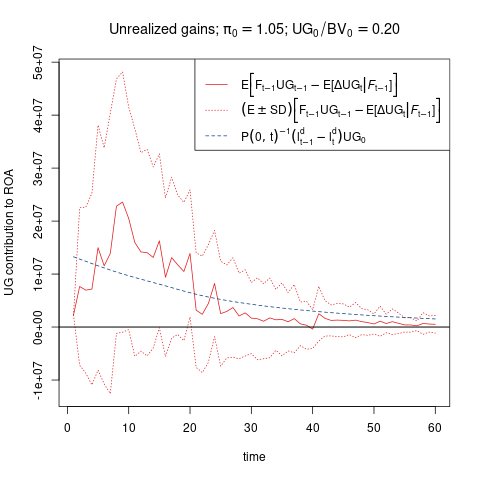}

\caption{Numerical evidence for the second part of Assumption~\ref{ass6:ROA}. The expected values and standard deviations are calculated by means of the numerical model described in Section~\ref{sec:numALM}, and the dashed line corresponds to the assumption.} 
\label{fig:ass6ROA2}
\end{figure}


Invoking the above assumptions \ref{ass:runoff}, \ref{ass6:ROA},  \ref{ass:tech}, \ref{ass2:sigma} and \ref{ass3:SF}, we define
\begin{align}
\label{e:gshat}
    \widehat{gs}_t 
    :&= 
    F_{t-1} E[BV_{t-1}] 
    + P(0,t)^{-1}(l_{t-1}^d - l_t^d)UG_0
    - \rho_t V_{t-1} + \gamma_t LP_{t-1} \\
    \notag 
    &=
    \Big(
    F_{t-1} 
    + P(0,t)^{-1}\frac{l_{t-1}^d-l_t^d}{l_{t-1}^h}\frac{UG_0}{(1+\vartheta)LP_0}
    -\frac{(1-\sigma)(\rho_t-\gamma_t)}{1+\vartheta}
    \Big) (1+\vartheta)l_{t-1}^h LP_0 
\end{align}
to be used as a model for $gs_t$ in the estimation of $COG$. Due to its relevance we state this conclusion as an equation: 
\begin{equation}\label{e:ass8}
 gs_t = \widehat{gs}_t
\end{equation}
with the understanding that this equality is to be understood in the approximative sense. Since we aim to use $\widehat{gs}_t$ in order to calculate an upper bound, $\widehat{COG}$, for the cost of guarantee, $COG$,  equation~\eqref{e:ass8} is justified if 
\begin{equation}\label{e:ass8test}
 E\Big[\sum_{t=1}^T B_t^{-1} gs_t^- \Big] 
 \le 
 E\Big[\sum_{t=1}^T B_t^{-1} \widehat{gs}_t^- \Big] .
\end{equation}
Empirical evidence for assertions \eqref{e:ass8} and \eqref{e:ass8test} is contained in Figures~\ref{fig:ass8} and \ref{fig:ass8test}.

\begin{figure}[ht]
\includegraphics[width=5.2cm]{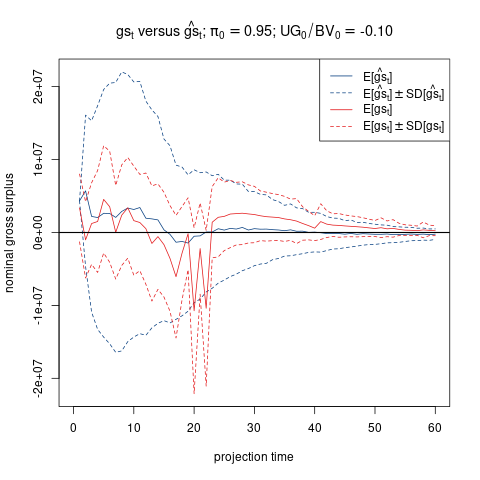}
\includegraphics[width=5.2cm]{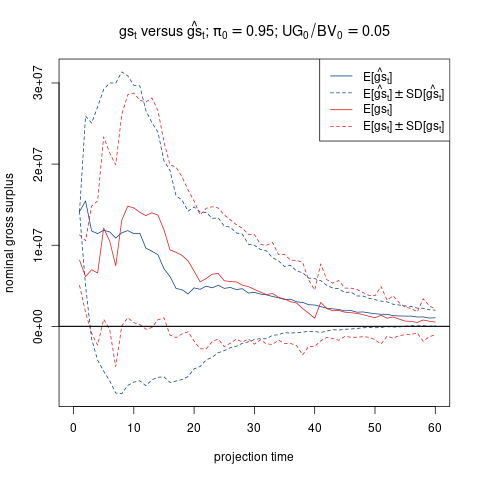}
\includegraphics[width=5.2cm]{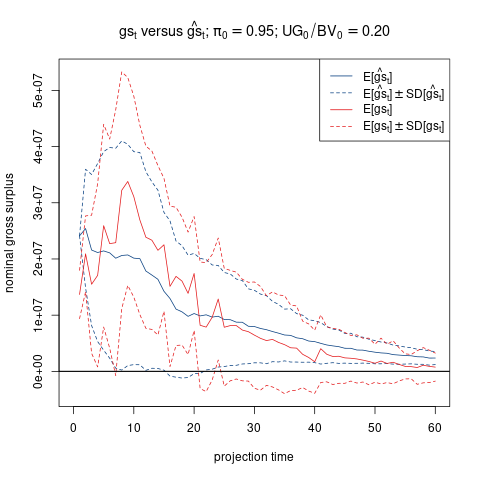}

\includegraphics[width=5.2cm]{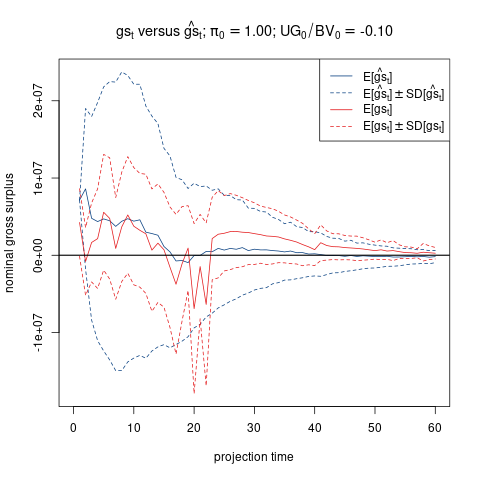}
\includegraphics[width=5.2cm]{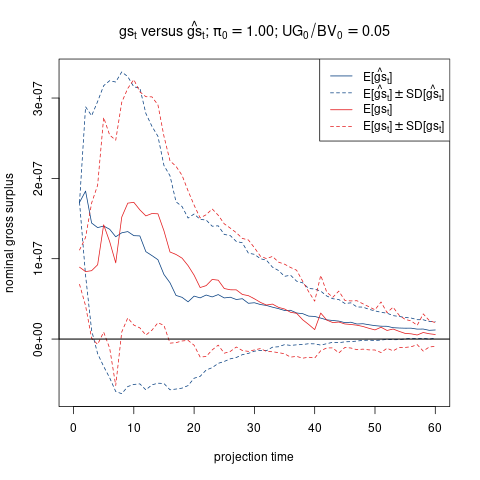}
\includegraphics[width=5.2cm]{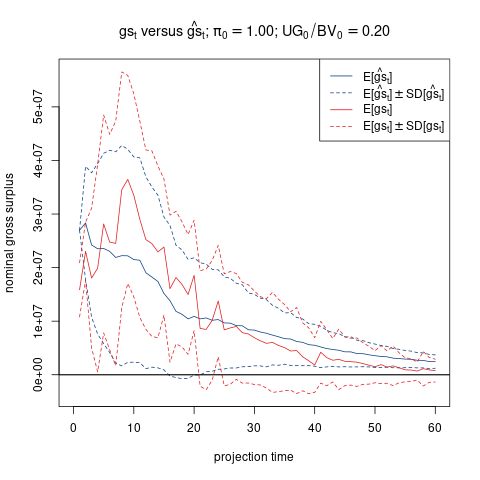}

\includegraphics[width=5.2cm]{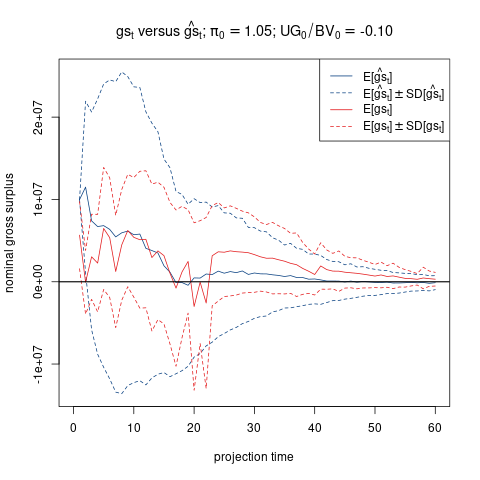}
\includegraphics[width=5.2cm]{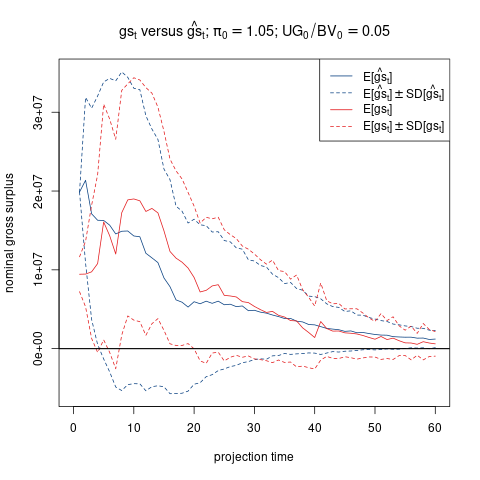}
\includegraphics[width=5.2cm]{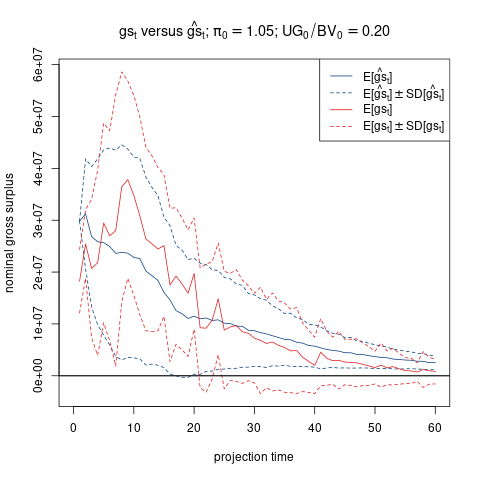}

\caption{Numerical -ex post- evidence for the validity of the approximating model \eqref{e:gshat}. Values are nominal, i.e.\ without discounting. It can be observed that the red lines, corresponding to expectation and standard deviation of the empirically modelled gross surplus (according to Section~\ref{sec:numALM}) have a tendency to remain above the $x$-axis. This behaviour is due to the management rules of Section~\ref{sec:MR} which try to avoid shareholder capital injections in order to reduce the cost of guarantee, $COG$, defined in \eqref{e:cog}.  }
\label{fig:ass8}
\end{figure}

\begin{figure}[ht]
\includegraphics[width=5.2cm]{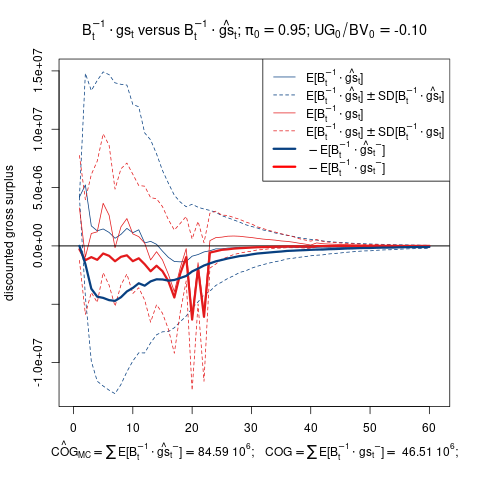}
\includegraphics[width=5.2cm]{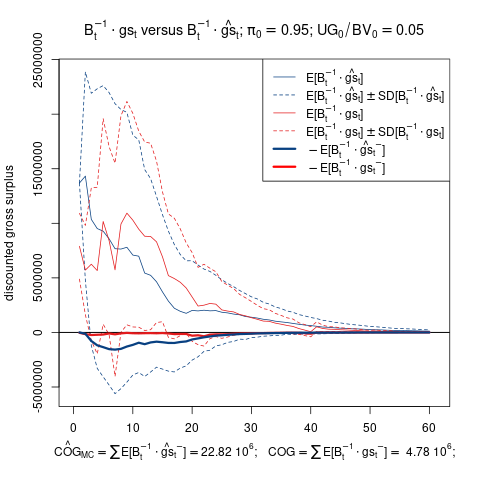}
\includegraphics[width=5.2cm]{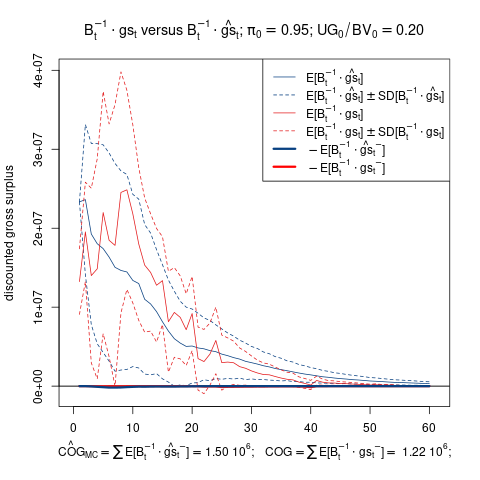}

\includegraphics[width=5.2cm]{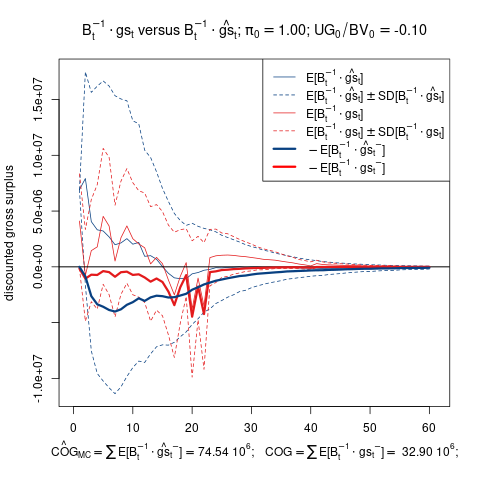}
\includegraphics[width=5.2cm]{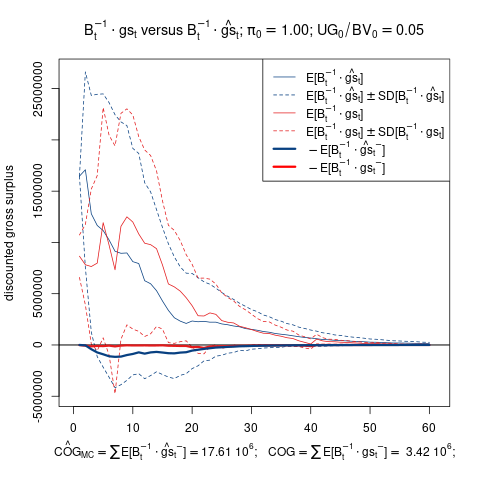}
\includegraphics[width=5.2cm]{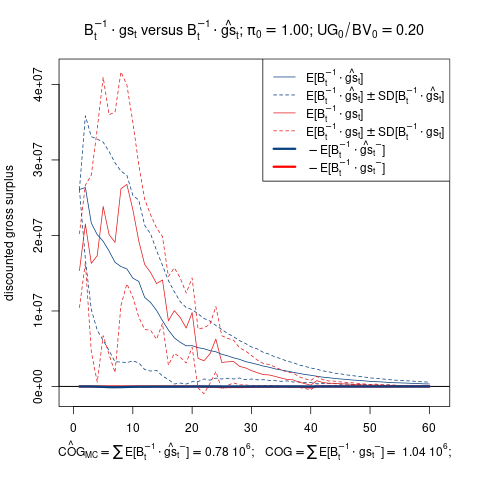}

\includegraphics[width=5.2cm]{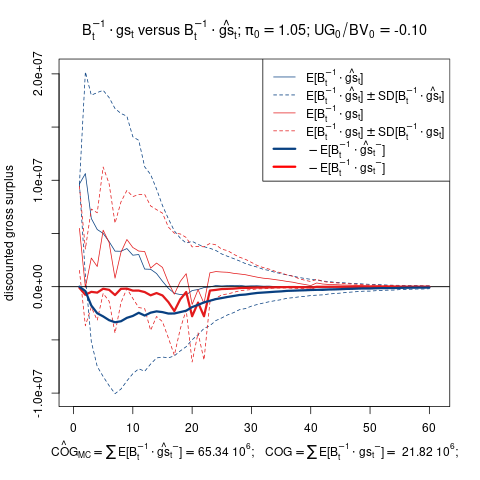}
\includegraphics[width=5.2cm]{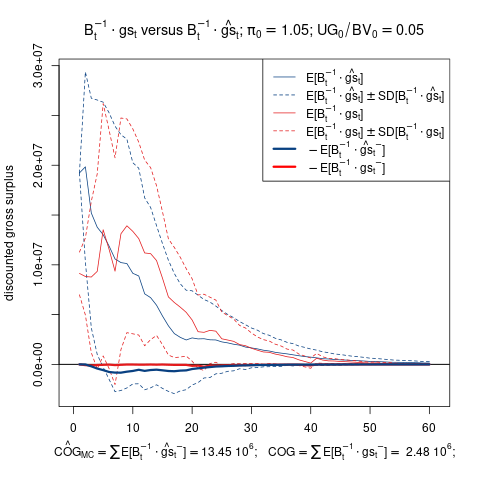}
\includegraphics[width=5.2cm]{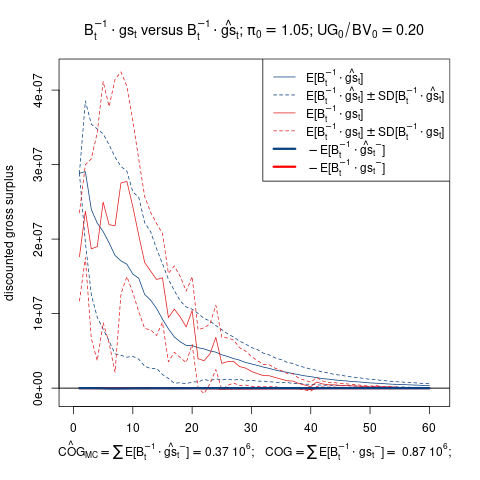}

\caption{Numerical -ex post- evidence for the validity of the approximating model \eqref{e:gshat}. Values are discounted along the appropriate interest rate scenarios. It can be observed that the red lines, corresponding to expectation and standard deviation of the empirically modelled gross surplus (according to Section~\ref{sec:numALM}) have a tendency to remain above the $x$-axis. This behaviour is due to the management rules of Section~\ref{sec:MR} which try to avoid shareholder capital injections in order to reduce the cost of guarantee, $COG$, defined in \eqref{e:cog}. The empirical values $COG$, calculated numerically (according to Section~\ref{sec:numALM}), and $\widehat{COG}_{MC}$, calculated also numerically (according to Section~\ref{sec:numALM}) but with respect to the simplified model \eqref{e:gshat} are displayed below the plots. 
} 
\label{fig:ass8test}
\end{figure}

\begin{remark}
The above assumptions are a streamlined version of those in \cite{GH22}, enhanced with numerical tests. Concretely, Assumption~\ref{ass:runoff} corresponds to \cite[Ass.~4.1 \& 4.2]{GH22}; Assumption~\ref{ass2:sigma} corresponds to \cite[Ass.~4.3]{GH22}; Assumption~\ref{ass3:SF} corresponds to \cite[Ass.~4.4]{GH22}; Assumption~\ref{ass:surr} corresponds to \cite[Ass.~4.5]{GH22}; 
Assumption~\ref{ass:cov} has similar elements as \cite[Ass.~4.10]{GH22}, but the assumption is not as strong (\cite[Ass.~4.10]{GH22} is not supported by our numerical results as the dotted line in Figure~\ref{fig:ass_cov} shows); 
Assumption~\ref{ass:tech} corresponds to \cite[Ass.~4.8]{GH22}; 
Assumption~\ref{ass6:ROA} corresponds to \cite[Ass.~4.9]{GH22}. 
The necessity of \cite[Ass.~4.6 \& 4.7]{GH22} is no longer given since we have simplified the calculation of the lower bound estimate.  
\end{remark}

\section{Estimation  of future discretionary benefits}\label{sec:est}
Based on the assumptions which are stated and discussed together with numerical evidence in the previous Section~\ref{sec:assump} we  shall now derive estimates, $\widehat{LB}$ and $\widehat{UB}$, for lower and upper bounds of the $FDB$. The merit of these bounds lies in the fact that they can be computed a priori without the need for Monte Carlo methods. Concretely, we shall estimate the terms $I$, $II$, $III$ and $COG$ from the representation~\eqref{e:fdb-rep}, and the corresponding estimators will be denoted by $\widehat{(\cdot)}$.
 
\subsection{Estimating $I$}
The portfolio is in run-off, and assumption~\ref{ass:runoff} implies correspondingly the trivial estimate
\begin{equation}
\label{e:Ihat}
    \widehat{I} = 0
\end{equation}
for $I$.

\subsection{Estimating $II$}

Assumptions~\ref{ass2:sigma}, \ref{ass:surr} and \ref{ass:cov} imply that an upper bound for $II$ can be estimated as
\begin{align}
\label{e:II_hat}
    \widehat{II} 
    := 
    (1-gph)\sum_{t=2}^T
    P(0,t) \,
    \gamma_t^{+}\,
    \sigma^{DB}_t \, l_{t-1}^h \, LP_0.
\end{align}

\subsection{Estimating $III$} 
Invoking assumptions~\ref{ass2:sigma}, \ref{ass3:SF} and ~\ref{ass:cov} we obtain the estimate
\begin{align}
\notag 
 III 
 &\le 2(1-gph)\sum_{t=1}^T\Big(P(0,t-1) - P(0,t)\Big)E\Big[DB_{t-1}+SF_{t-1}\Big]  \\ 
 \notag 
 &\le 
 2(1-gph)\sum_{t=1}^T\Big(P(0,t-1) - P(0,t)\Big)\Big(\sigma^{DB}_t + \theta\Big)l_{t-1}^h LP_0 \\ 
 &=:
 \widehat{III}
 \label{e:IIIhat}
\end{align}
as an upper bound of $III$. 

\subsection{Estimating $COG$}
Making use of assumptions \ref{ass:runoff}, \ref{ass6:ROA},  \ref{ass:tech}, \ref{ass2:sigma} and \ref{ass3:SF}, we model $gs_t$ by $\widehat{gs}_t$ where the latter is defined in \eqref{e:gshat}. 
Correspondingly the cost of guarantee, $COG$, defined in equation~\eqref{e:cog} is estimated as 
\begin{equation}\label{e:COGhat}
  \widehat{COG}
  = E\Big[\sum_{t=1}^T B_t^{-1} \widehat{gs}_t^-  \Big] 
  = \sum_{t=1}^T \mathcal{O}_t^- (1+\vartheta)l_{t-1}^h LP_0  
\end{equation}
and
\begin{equation}
    \label{e:O}
 \mathcal{O}_t^{-}
 := E\Big[B_t^{-1}\Big(
   F_{t-1} 
    + P(0,t)^{-1}\frac{l_{t-1}^d-l_t^d}{l_{t-1}^h}\frac{UG_0}{(1+\vartheta)LP_0}
    - \frac{(1-\sigma_t)(\rho_t-\gamma_t)}{1+\vartheta}
 \Big)^{-}\Big]
\end{equation} \label{ref: O_s}
is the value at $0$ of the floorlet (i.e., negative caplet) with maturity $t-1$ and payment $\min(F_{t-1} - k_t, 0)$ at settlement date $t$,  where 
the strike is given by 
\begin{equation}
    \label{e:strike}
    k_t
    := 
    - P(0,t)^{-1}\frac{l_{t-1}^d-l_t^d}{l_{t-1}^h}\frac{UG_0}{(1+\vartheta)LP_0}
    + \frac{(1-\sigma_t)(\rho_t-\gamma_t)}{1+\vartheta}.
\end{equation} \label{ref: k_s}
In the normal model this value is given by the Black formula
 (\cite{Black76,BM06})
\begin{equation} 
\label{e:Black} 
    \mathcal{O}_t^{-} 
    = 
    P(0,t)\cdot\Big(
     \pm ( F_{t-1}^0 - k_t )\Phi( - \kappa_t) + \textup{IV}_t \sqrt{t}\phi( - \kappa_t)
     \Big)
\end{equation}
where $\Phi$ and $\phi$ are the normal cumulative distribution and density functions, respectively.  
Further,
\[
  \kappa_t  
  := 
  \frac{ F_{t-1}^0 - k_t }
   { \textup{IV}_t\sqrt{t} } 
\]
where $F_{t-1}^0 = P(0,t-1) - P(0,t)$ is the forward rate prevailing at time $0$, 
and $\textup{IV}_t$\label{ref: IV_s} is the caplet implied volatility known from market data.

\subsection{Estimating $FDB$}
These estimates yield a lower bound, $\widehat{LB}$, and an upper bound, $\widehat{UB}$, for $FDB$: 
\begin{equation}
    \label{e:est-int}
    \widehat{LB} 
    \le FDB \le \widehat{UB}
\end{equation}
where 
\begin{align}
    \label{e:LB}
    \widehat{LB}
    &:= 
    FDB_0
    - \widehat{II}
    - \widehat{III}
    \\
    \label{e:UB}
    \widehat{UB}
    &:=
   FDB_0 
   + gph\cdot\widehat{COG}.
\end{align}
If the difference $\eps = (\widehat{UB}-\widehat{LB})/2$ is sufficiently small (e.g., in comparison to $MV_0$), then 
\begin{equation}\label{FDBhat}
\widehat{FDB} 
= \frac{\widehat{LB}+\widehat{UB}}{2}
= SF_0 + gph\Big(LP_0+UG_0-GB\Big) 
  + \frac{
    gph\cdot\widehat{COG}  
     - \widehat{II}
    - \widehat{III} 
     }{2} 
\end{equation}
may be used as a reasonable estimator for $FDB$ with maximal error given by $\eps$.

\subsection{Comparison to numerically calculated values}\label{sec:num_stu}
The purpose of this section is to apply the model of Section~\ref{sec:numALM} to a realistic life insurance portfolio, and to compare the results with the estimates involved in calculating \eqref{FDBhat}. The life insurance portfolio that we utilize to this end consists of anonymous, aggregated and randomized asset and liability data from multiple life insurance companies. Moreover, the portfolio is perturbed initially according to \eqref{e:UGstart} and \eqref{e:pi} such that we obtain nine different cases. These correspond to quite different relations between and assets and liabilities, and are therefore considered to be representative  for a wide variety of possible company data and economic environments. Nevertheless, the conclusions to be drawn from this sections cannot be definite but only those that follow from a specific numerical study. The results from this study are summarized in Table~\ref{tab:num_res}.

Given that the assumptions which go into the derivation of the estimation formula~\eqref{FDBhat} are verified numerically in Section~\ref{sec:assump} it is to be expected that the empirical and estimated values in Table~\ref{tab:num_res} are consistent. This is indeed the case, and it can be observed that the values agree quite well in most cases, particularly when viewed relative to the initial market value $MV_0 = LP_0+SF_0+UG_0$. 

\begin{table}[ht] 
\centering
\begin{tabular}{rrrrrrrrrrr}
  \hline\hline 
$\pi_0$ & 0.95 & 0.95 & 0.95 & 1.00 & 1.00 & 1.00 & 1.05 & 1.05 & 1.05 & \eqref{e:pi} \\ 
  $UG_0/BV_0$ & -0.10 & 0.05 & 0.20 & -0.10 & 0.05 & 0.20 & -0.10 & 0.05 & 0.20 & \eqref{e:UGstart}\\
  \hline 
  $LP_0$ & 1000.00 & 1000.00 & 1000.00 & 1000.00 & 1000.00 & 1000.00 & 1000.00 & 1000.00 & 1000.00 & Sec.~\ref{sec:li} \\ 
  $SF_0$ & 20.00 & 20.00 & 20.00 & 20.00 & 20.00 & 20.00 & 20.00 & 20.00 & 20.00 & Sec.~\ref{sec:li}\\ 
  $MV_0$ & 918.00 & 1071.00 & 1224.00 & 918.00 & 1071.00 & 1224.00 & 918.00 & 1071.00 & 1224.00 & Sec.~\ref{sec:li}\\ 
  $UG_0$ & -102.00 & 51.00 & 204.00 & -102.00 & 51.00 & 204.00 & -102.00 & 51.00 & 204.00 & Sec.~\ref{sec:li}\\ 
  $FDB_{0}$ & 9.72 & 133.58 & 257.43 & 29.69 & 153.54 & 277.40 & 49.66 & 173.51 & 297.36 & Thm.~\ref{thm:rep} \\ 
  $GB$ & 910.70 & 910.70 & 910.70 & 886.03 & 886.03 & 886.03 & 861.37 & 861.37 & 861.37 & \eqref{e:GB} \\ 
  \hline 
  $E[B_T^{-1}MV_T]$ & 0.10 & 0.20 & 0.29 & 0.11 & 0.22 & 0.31 & 0.13 & 0.24 & 0.33 & Sec.~\ref{sec:li}\\ 
  $TAX$ & 2.51 & 9.03 & 17.73 & 3.18 & 10.40 & 19.21 & 4.05 & 11.69 & 20.62 & Sec.~\ref{sec:li}\\ 
  $VIF$ & -39.46 & 20.59 & 48.61 & -23.96 & 25.81 & 52.93 & -10.43 & 30.38 & 57.06 & Sec.~\ref{sec:li}\\
  \hline 
  $SHG$ & 7.06 & 25.37 & 49.83 & 8.94 & 29.23 & 53.97 & 11.39 & 32.86 & 57.93 & Sec.~\ref{sec:li}\\ 
  $COG$ 
  & 46.51 & 4.78 & 1.22 
  & 32.90 & 3.42 & 1.04 
  & 21.82 & 2.48 & 0.87 & \eqref{e:cog} \\ 
  $\widehat{COG}_{MC}$ 
  & 84.59 & 22.82 & 1.50 
  & 74.54 & 17.61 & 0.78
  & 65.34 & 13.45 & 0.37 &  Fig.~\ref{fig:ass8test}
  \\ 
  $\widehat{COG}$ & 74.62 & 23.57 & 5.87 & 64.61 & 19.33 & 4.72 & 55.84 & 15.95 & 3.85 & \eqref{e:COGhat} \\ 
  \hline 
  $I$ & 0.09 & 0.19 & 0.27 & 0.10 & 0.20 & 0.29 & 0.12 & 0.22 & 0.30 & Thm.~\ref{thm:rep} \\ 
  $\widehat{I}$ & 0.00 & 0.00 & 0.00 & 0.00 & 0.00 & 0.00 & 0.00 & 0.00 & 0.00 & \eqref{e:Ihat} \\ 
  $II$ & 0.04 & 0.09 & 0.16 & 0.04 & 0.10 & 0.17 & 0.05 & 0.11 & 0.18 & Thm.~\ref{thm:rep} \\ 
  $\widehat{II}$ & 0.01 & 0.01 & 0.01 & 0.26 & 0.77 & 1.39 & 0.56 & 1.92 & 3.31  & \eqref{e:II_hat} \\ 
  $III$ & 3.30 & 6.77 & 11.44 & 3.69 & 7.53 & 12.22 & 4.15 & 8.21 & 12.96 & Thm.~\ref{thm:rep} \\ 
  $\widehat{III}$ & 6.33 & 12.88 & 23.33 & 6.40 & 14.94 & 25.82 & 6.48 & 17.16 & 28.46  & \eqref{e:IIIhat} \\ 
  \hline 
  $LT_{CF}$ & 1.04 & 0.44 & 0.63 & 0.81 & 0.32 & -0.16 & -0.62 & 1.82 & 0.07 \\ 
  $LT_{rep}$ & 0.20 & 0.08 & 0.12 & 0.15 & 0.06 & -0.03 & -0.12 & 0.35 & 0.01 \\ 
  $FDB_{CF}^{mce}$ & 0.30 & 0.73 & 0.91 & 0.38 & 0.76 & 0.93 & 0.46 & 0.77 & 0.94 \\ 
  $FDB_{rep}^{mce}$ & 0.38 & 0.12 & 0.10 & 0.33 & 0.11 & 0.10 & 0.27 & 0.10 & 0.10 
  \\ 
  \vspace{-2ex}
  \\
  \hline 
  \multicolumn{1}{|c}{$FDB_{CF}$} & 43.11 & 130.04 & 246.04 & 51.83 & 148.22 & 265.69 & 63.50 & 165.50 &  {284.56} 
  &  \multicolumn{1}{c|}{\eqref{e:FDB}} \\ 
  \multicolumn{1}{|c}{$FDB_{rep}$} & 43.95 & 130.40 & 246.55 & 52.48 & 148.48 & 265.56 & 62.99 & 166.98 &  {284.62}
  &  \multicolumn{1}{c|}{\eqref{e:fdb-rep}} \\ 
  \multicolumn{1}{|c}{$\widehat{FDB}$} & 36.76 & 136.67 & 248.13 & 52.51 & 153.51 & 265.70 & 68.74 & 170.43 & {283.03} 
  &  \multicolumn{1}{c|}{\eqref{FDBhat}} \\ 
  \hline 
  \\ 
  \vspace{-4.5ex}
  \\
  $\widehat{LB}$ & 3.39 & 120.68 & 234.09 & 23.03 & 137.83 & 250.19 & 42.62 & 154.43 & 265.59 & \eqref{e:LB} \\ 
  $\widehat{UB}$ & 70.13 & 152.65 & 262.18 & 81.99 & 169.19 & 281.22 & 94.86 & 186.42 & 300.48 & \eqref{e:UB} \\ 
  $\delta/MV_0$ &
  -0.73\,\%& 0.60\,\% & 0.13\,\% & 0.14\,\% & 0.48\,\% & 0.02\,\% & 0.56\,\% & 0.32\,\% & -0.13\,\%
  \\ 
  $\eps/MV_0$ &
  3.64\,\% & 1.49\,\% & 1.15\,\% & 3.21\,\% & 1.46\,\% & 1.27\,\% & 2.85\,\% & 1.49\,\% & 1.42\,\%
  \\ 
   \hline\hline 
\end{tabular}
\vspace{2ex} 
\caption{Values in $10^6$ units of currency. $FDB_{CF}$ denotes the numerically calculated quantity corresponding to the cash-flow based definition \eqref{FDBhat}, and $FDB_{rep}$ denotes the numerically calculated quantity corresponding to the representation \eqref{e:fdb-rep}.
The leakage test quantities, $LT_{CF}$ and $LT_{rep}$, are obtained by subtracting the numerically calculated right hand side from the left hand side of equation~\eqref{e:noL} where $FDB$ is given by $FDB_{CF}$ and $FDB_{rep}$, respectively. Similarly, $FDB_{CF}^{mce}$ and $FDB_{rep}^{mce}$ are the Monte Carlo errors corresponding to the numerically calculated values $FDB_{CF}$ and $FDB_{rep}$, respectively. Finally, $\delta = \widehat{FDB} - FDB_{rep}$ and $\eps = (\widehat{UB}-\widehat{LB})/2$ are expressed in percent relative to $MV_0$, and the estimation is considered successful if $|\delta|<\eps$, which holds in all cases. The last column provides references to definitions for most quantities. 
} 
 \label{tab:num_res}
\end{table}

Table~\ref{tab:num_res} is compartmentalized as follows. The top two lines consist of the parameters $\pi_0$ and $UG_0/BV_0$ which serve to vary the relationship between assets and liabilities. The quantities $LP_0$, $SF_0$, $MV_0$ and $UG_0$ are starting values, and $FDB_0$ and $GB$ are calculated deterministically. 

The value $E[B_T^{-1}MV_T]$ corresponds to the expected discounted final market value of the asset portfolio. Since the company is modeled with respect to run-off assumptions, this quantity is expected to be small, as is verified here. In fact, smallness of this term is a necessary condition: if it is not verified then the projection horizon has to be increased. The quantity $TAX$ represents the value of corporate tax payments. The term $VIF$ corresponds to the value of in-force business. This is the value of the life insurance portfolio from the shareholder's perspective. The \emph{market consistent embedded value} that is often disclosed in financial reports is defined accordingly as $MCEV = \revise{FC_0} + VIF$\revise{, where $FC_0$ is free capital at time 0.}. It is the shareholder's objective to maximize $VIF$, cf.\ the control problem in Section~\ref{sec:con}. It can be observed in Table~\ref{tab:num_res} that this value is subject to quite a significant variation depending on the initial conditions. 

The value of in-force business is defined in Section~\ref{sec:li} as the difference of shareholder gains and cost of guarantee, $VIF = SHG - COG$. In Table~\ref{tab:num_res}, $SHG$ and $COG$ are calculated numerically by means of the model set-up in Section~\ref{sec:numALM}. The values $\widehat{COG}_{MC}$ and $\widehat{COG}$ correspond to the Monte-Carlo valuation of \eqref{e:COGhat} and the valuation of \eqref{e:COGhat} by means of the Black formula~\eqref{e:Black}, respectively. If the model were calibrated perfectly we would have $\widehat{COG}_{MC} = \widehat{COG}$.  For our purpose of estimating the upper bound it is relevant that $\widehat{COG} \ge COG$. This holds for all cases. In fact, we observe an over-estimation. This over-estimation is related to Figure~\ref{fig:ass8test} where it can be seen graphically that $\widehat{COG}_{MC}$, and therefore also $\widehat{COG}$, does not take the $COG$-reducing effect of management rules into account. When $COG$ and $\widehat{COG}_{MC}$ are very low, which occurs in every third column corresponding to a high fraction $UG_0/BV_0$, we observe that $\widehat{COG}_{MC}$ is much smaller than $\widehat{COG}$ indicating imperfect model calibration. In these cases it may also happen that $\widehat{COG}_{MC} < COG$, however since this occurs only for negligible values of $COG$ this observation does not impede the applicability of the upper bound formula~\eqref{e:UB}. 

The terms $I$, $II$ and $III$ in Table~\ref{tab:num_res} are calculated numerically according to their defining formulas in Theorem~\ref{thm:rep}. It can be observed that the estimates $\widehat{I} = 0$ and $\widehat{II}$ are quite sharp, and that $\widehat{II}$ might, at least according to this study, equally well have been set to $0$. Indeed, $\widehat{II}$ corresponds to shareholder gains on surrender penalties stemming from future profit declarations. This involves surrender probabilities and penalty factors which typically decline to $0$ towards the maturity of a contract. From these considerations it can already be deduced that $\widehat{II}$ should be expected to be quite small in comparison to $FDB$. The third estimate, $\widehat{III}$, shows an over-estimation. This can be related to Assumption~\ref{fig:ass_cov} where the relative covariance is estimated by $1$ while Figure~\ref{fig:ass_cov} shows that this covariance might have been estimated by $0$. If this were the employed estimate then the factor $2$ in \eqref{e:IIIhat} would be replaced by $1$, and $\widehat{III}$ and $III$ in Table~\ref{tab:num_res} would agree quite well. However, we do not strive to have the sharpest estimates $\widehat{LB}$ and $\widehat{UB}$, and therefore prefer these more generous -but also more stable- assumptions. 

The leakage and Monte Carlo error results in Table~\ref{tab:num_res} show the usefulness of representation~\eqref{e:fdb-rep} in numerical calculations. Both are much smaller when the representation formula is used to compute the future discretionary benefits. This is not surprising since $FDB_0$ depends only on the guaranteed benefits but is independent from stochastic scenarios, and this already determines $FDB_{rep}$ to a large extent -- as can be verified by reading off the other terms from Table~\ref{tab:num_res}. 

The boxed results containing $FDB_{CF}$, $FDB_{rep}$ and $\widehat{FDB}$ contain the core of Table~\ref{tab:num_res}. First of all, notice that  $FDB_{CF}$ and $FDB_{rep}$ coincide up to their respective Monte Carlo errors, as it should be according to Theorem~\ref{thm:rep}. Moreover, the difference $\delta = \widehat{FDB} - FDB_{rep}$ is quite small compared to $MV_0$, and satisfies $|\delta|<\eps$ in all cases. Notice, finally, that $FDB_0$ while always in the interval from $\widehat{LB}$ to $\widehat{UB}$ is not generally a good estimator for $FDB$ because it would neglect the effect of the guarantee cost $COG$.









\section{Application to publicly available reporting data} 

The purpose of this section is to apply the estimation formulae~\eqref{e:LB} and \eqref{e:UB} for lower and upper bound to publicly available data from a real life insurer (Allianz Lebensversicherung~AG in Germany), and to compare the results to the  value, $FDB_{CO}$,  that has been derived by the company. In Section~\ref{sec:All_data} we collect the necessary data, the estimates and a discussion are then provided in Section~\ref{sec:All_est}. These estimations are carried out for the accounting years 2017-\revise{2023}, and are therefore subject to quite different circumstances: low interest rate environment and positive unrealized gains for 2017-\revise{2021}, and medium-high yield curve in 2022 accompanied by negative unrealized gains.

\subsection{Allianz~Lebensversicherungs-AG: publicly reported values} \label{sec:All_data}
Formulae~\eqref{e:LB} and \eqref{e:UB} hinge on the knowledge of specific balance sheet items. These items are contained in publicly available data, and provided for further reference in Table~\ref{table:values}.  The respective sources are listed in Table~\ref{tbl:source}.

\begin{table}[ht]
\centering
\begin{tabular}{lrrrrrrr}
Quantity                      
& \multicolumn{1}{l}{2017} 
& \multicolumn{1}{c}{2018} 
& \multicolumn{1}{c}{2019}  
& \multicolumn{1}{c}{2020}
& \multicolumn{1}{c}{2021}
& \multicolumn{1}{c}{2022}
& \multicolumn{1}{c}{\revise{2023}}\\ 
\hline
$L_0 = LP_0+SF_0$                     
& 189.8 & 201.2 & 219.6 & 235.1  & 246.1 & 247.6 
& 248.3
\\
$UG_0$ 
& 41.4 & 32.8 & 54.0 & 66.2 & 55.1 & -16.3 
& -9.4
\\
$SF_0$ 
& 10.5 & 11.0 & 11.5 & 12.3 & 12.5 & 11.6 
& 11.5
\\
Solvency~II value of~$SF_0$ 
& 10.4 & 10.5 & 11.3 & 12 & 12.2 & 9.0 
& 10.9
\\
$GB$ 
& 154.1 & 158.8 & 195.2 & 223.9 & 228.4 & 163.6 
& 174.2
\\
$FDB_{CO}$ 
& 48.6 & 46.2 & 47.4 & 44.7 & 37.6 &39.5 
& 36.9
\\
$GC$ 
&3.3 & 3.2 & 4.2  & 3.7 & 3.5 & 3.8 
& 4.3
\\
$IRR$ 
& 7.9 & 12.5 & 14.8 & 17.4 & 19.6 & 19 
& 18.3
\\
\hline
\end{tabular}
\vspace{2ex} 
\caption{Allianz Lebensversicherungs-AG: public data for 2017-\revise{2023}, values are in billion euros. Sources are contained in Table~\ref{tbl:source}.}
\label{table:values}
\end{table}

The local GAAP value, $L_0$, of life insurance with profit participation is adjusted due to a necessary regrouping of business. This step is explained respectively on the first of the two cited page numbers of the SCR report \cite{SFCR} for each accounting year. 
The value of $UG_0$ is scaled to $L_0$, which is in line with the general assumption \eqref{e:BV = LP + SF}. The reason behind this scaling is that according to \cite[§~3]{MindZV} only the fraction of the capital gains, corresponding to the assets scaled to cover the average value of liabilities in the accounting year under consideration, contribute to the gross surplus. Moreover, the value $GB$ that is used to calculate $FDB_0$ is augmented by the going concern reserve $GC$. This is because this reserve is a cash-flow that leaves the model (as explained in the quoted source in Table~\ref{tbl:source}) but this cash-flow is not accounted for in our no-leakage principle~\ref{e:noL}, and we do not know to which cash-flow in the comprehensive list recorded in the delegated act \cite[Article~28]{L2} this going concern reserve should be attributed to.

\begin{table}[ht]
\begin{minipage}{\textwidth}
\tiny{
\centering
\begin{tabular}{lrrrrrrr}
Quantity                
& Source 2017   & Source 2018  & Source 2019  
& Source 2020 & Source 2021 & Source 2022 
& \revise{Source 2023}
\\ 
\hline
$L_0$\footnote{Versicherung mit Überschussbeteiligung} 
& \cite[p.~46,~52]{SFCR} & \cite[p.~42,~46]{SFCR} 
&\cite[p.~42,~46]{SFCR} & \cite[p.~50,~56]{SFCR}
& \cite[p.~49,~55]{SFCR} & \cite[p.~55,~55]{SFCR} 
& \cite[p.~54,~54]{SFCR} 
\\
$UG_0$\footnote{Summe der in die \"Uberschussbeteiligung   einzubeziehenden Kapitalanlagen}                      
 &\cite[p.~46]{GB} &\cite[p.~42]{GB} & \cite[p.~46]{GB}  &\cite[p.~49]{GB} &\cite[p.~50]{GB} & \cite[p.~50]{GB} 
 & \cite[p.~51]{GB} 
 \\
$SF_0$\footnote{Rückstellung für Beitragsrückerstattung abzüglich
 festgelegte, aber noch nicht zugeteilte Teile} 
 &\cite[p.~55]{GB} & \cite[p.~51]{GB} & \cite[p.~55]{GB} 
 &\cite[p.~30]{GB} & \cite[p.~60]{GB} & \cite[p.~59]{GB}
 & \cite[p.~60]{GB}\\

SII value of~$SF_0$\footnote{Überschussfonds} 
& \cite[p.~52]{SFCR} & \cite[p.~46]{SFCR} & \cite[p.~46]{SFCR} 
& \cite[p.~56]{SFCR} & \cite[p.~55]{SFCR} & \cite[p.~55]{SFCR} 
& \cite[p.~54]{SFCR} 
\\
$GB$\footnote{Bester Schätzwert: Wert für garantierte Leistungen} 
& \cite[p.~46]{SFCR} & \cite[p.~42]{SFCR}  & \cite[p.~42]{SFCR} 
& \cite[p.~50]{SFCR} & \cite[p.~49]{SFCR} & \cite[p.~49]{SFCR}
& \cite[p.~48]{SFCR}
\\
$FDB_{CO}$\footnote{Bester Schätzwert: zukünftige Überschussbeteiligung} 
& \cite[p.~46]{SFCR} &  \cite[p.~42]{SFCR} &  \cite[p.~42]{SFCR}      & \cite[p.~50]{SFCR}&  \cite[p.~49]{SFCR} 
&\cite[p.~49]{SFCR} & \cite[p.~48]{SFCR}
\\
$GC$\footnote{Going concern reserve} 
& \cite[p.~52]{SFCR} & \cite[p.~46]{SFCR} & \cite[p.~52]{SFCR}    & \cite[p.~49]{SFCR}  & \cite[p.~50]{SFCR} & \cite[p.~52]{SFCR}
 & \cite[p.~54]{SFCR}
\\
$IRR$\footnote{Interest rate reserve (Zinszusatzrückstellung)}
& \cite[p.~51]{SFCR} &\cite[p.~50]{SFCR} &\cite[p.~51]{SFCR} 
& \cite[p.~51]{SFCR} & \cite[p.~50]{SFCR} &   \cite[p.~59]{GB} 
& \cite[p.~60]{GB}
\\
\hline
\end{tabular}
\vspace{2ex} 
\caption{Allianz Lebensversicherungs-AG. Publicly available sources for the data in Table~\ref{table:values}. The source is given as a web-page where the Financial Statement~\cite{GB} and SFC Report~\cite{SFCR} can be accessed for the indicated years, and the page numbers refer to reports of the corresponding years.}
\label{tbl:source}
}
\end{minipage}
\end{table}

The relevant risk free rate providing the deterministic discount factors $P(0,t)$ and the initial forwards $F_t^0$ is published by EIOPA at \cite{RFR}. 

All other data, for instance technical interest rate and technical gains rate, are as in \cite[Section~6]{GH22}. In particular, the policyholder fraction of the positive gross surplus is $gph = 75.5\,\%$.

\subsection{Estimation}\label{sec:All_est}
The estimation formulae~\eqref{e:LB} and \eqref{e:UB} have been implemented in $R$, and evaluated on the data of Section~\ref{sec:All_data}. The results of this implementation are displayed in Table~\ref{tab:0}.\footnote{The $R$-code, together with the data, can be provided in executable form upon request to the authors to reproduce Table~\ref{tab:0}.}

The estimations in Table~\ref{tab:0} are subject to considerable uncertainty. In \cite{GH22} we have studied sensitivities of the formulas~\eqref{e:LB} and \eqref{e:UB} with respect to parameter choices concerning, e.g., technical interest rate, technical gains, or market implied volatility. It was observed that these sensitivities impact the estimation only marginally. This is due to the fact that $\widehat{FDB}$ is determined to a large extent already by $FDB_0$. However, $FDB_0$ also carries estimation uncertainty, although this may not be obvious at first sight. Indeed, in principle all the data needed to calculate $FDB_0$ are contained in Table~\ref{table:values}, apart from $gph$. In theory these objects are all clearly defined, but in practice we had to: adjust $L_0$ to reflect only life insurance with profit participation, rescale $UG_0$ to $L_0$, increase $GB$ by $GC$, and deduct $SF_0$ from $\widehat{FDB}$. These adjustments were in order so that all items can be interpreted correctly. However, due to incomplete information we cannot know if these were all the necessary adjustments. For example, it may be possible that the regrouping that leads to a rescaling of $L_0$ should also have induced a rescaling of $SF_0$, followed by an appropriate scaling of $UG_0$. The point here is that the calculation of $FDB_0$ is only seemingly trivial but involves in practice a lot of information about the various positions.  

\begin{table}[H]\phantom{X}
\begin{align*}
&
\begin{matrix}
  & 
  FDB_{CO}  & \widehat{FDB} & \widehat{LB} & \widehat{UB} & \eps & \delta & FDB_0 & \widehat{II} & \widehat{III} & \widehat{COG} \\ 
  \\
 2017 
 & 48.60 & 47.35 & 42.99 & 51.71 & 4.36 & -1.25 & 47.57 & 1.08 & 3.51 & 5.51 \\ 
  2018 
  & 46.20 & 45.82 & 40.46 & 51.17 & 5.35 & -0.38 & 45.76 & 1.17 & 4.13 & 7.21 \\ 
  2019 
  & 47.40 & 47.92 & 42.41 & 53.44 & 5.52 & 0.52 & 47.01 & 1.57 & 3.03 & 8.58 \\ 
  2020 
  & 44.74 & 48.44 & 41.96 & 54.91 & 6.47 & 3.70 & 46.16 & 1.88 & 2.32 & 11.66 \\ 
  \revise{2021} 
  & 37.61 &   43.69 & 37.43 & 49.95 & 6.26 & 6.08 & 42.58 & 1.89 &  3.25 & 9.83  
  \\ 
  \revise{2022} 
  & 39.50 &    37.25 & 30.97 & 43.53 & 6.28 & -2.25 & 39.12 & 1.34   & 6.81 & 5.88 \\ 
  \revise{2023} 
  & 36.83 &  35.29 & 28.91 & 41.67 & 6.38 & -1.54 & 36.61 & 1.39 & 6.31 & 6.75 \\ 
  \\
  \hline
  \\
 2017 
 & 21.02 & 20.48 & 18.59 & 22.37 & 1.89 & -0.54 & 20.58 & 0.47 &  1.52 & 2.38 \\ 
  2018 
  & 19.74 & 19.58 & 17.29 & 21.87 & 2.29 & -0.16 & 19.56 & 0.50 &  1.76 & 3.08 \\ 
  2019 
  & 17.32 & 17.52 & 15.50 & 19.53 & 2.02 & 0.19 & 17.18 & 0.57 &  1.11 & 3.14 \\ 
  2020 
  & 14.85 & 16.07 & 13.92 & 18.22 & 2.15 & 1.23 & 15.32 & 0.62 &   0.77 & 3.87 \\ 
  \revise{2021} 
  & 12.49 
  & 14.51 & 12.43 & 16.59 & 2.08 & 2.02 & 14.14 & 0.63 &  1.08 & 3.26 \\ 
 \revise{2022} 
 & 17.08 &   16.11 & 13.40 & 18.83 & 2.72 & -0.97 & 16.92 & 0.58 &  2.95 & 2.54 \\ 
  \revise{2023} 
  & 15.42 &   14.77 & 12.10 & 17.44 & 2.67 & -0.64 & 15.33 & 0.58 & 2.64 & 2.82 \\ 
   \end{matrix}
\end{align*}
\vspace{2ex}
\caption{Allianz Lebensversicherung~AG vs.\ estimated values. Numbers in the top are in billion Euro, numbers in the bottom table are in percent of $MV_0$. $FDB_{CO}$ as in Table~\ref{table:values}. Quantities with a hat are calculated according to Section~\ref{sec:est}. $\delta = \widehat{FDB} - FDB_{CO}$ and $\eps = (\widehat{UB}-\widehat{LB})/2$. \revise{The estimation is successful, i.e.\ $|\delta|<\eps$, in all cases.}
} 
\label{tab:0} 
\end{table}

\revise{It can be seen that the estimations in Table~\ref{tab:0} are successful, i.e.\ $|\delta|<\eps$, in all years. 
In this regard we note that the economic conditions in $2022$ are characterized by a steep rise in the term structure in comparison to the previous years. This leads to a negative value of unrealized gains as recorded in Table~\ref{table:values}. Nevertheless the estimations continue to hold also under the new economic circumstances in $2022$ and $2023$. Moreover, the estimation intervals are of comparable magnitude as can be observed from the relative values of $\eps$ in Table~\ref{tab:0}. 
}

\section{Conclusions}
In the present paper we have concerned ourselves with three strands of thought all connected to the market consistent valuation of realistic life insurance portfolios: (1) interest rate scenario  generation suitable for long term modelling; (2) numerical ALM modelling; (3) algebraic formulae to estimate numerical results. 

From a practical point of view we may draw the following conclusions concerning the calculation of best estimates or future discretionary benefits: 

\begin{enumerate}
\item 
Theorem~\ref{thm:mflmm} provides sufficient conditions for the existence and uniqueness of solutions to the mean-field Libor market model. If the MF-LMM is specified by \eqref{e:LMM-mf-var} these assumptions are straightforward to check. This allows to extend the Libor market model, which is very popular in  the insurance industry, to the mean-field setting, and thus reduce the probability of explosion (cf.\ \cite{MFLMM}). 
\item 
The representation formula~\eqref{e:fdb-rep} allows to represent the $FDB$ without any policyholder cash-flows. Since $FDB_0$ is known, one expects that this formula should lead to a reduced Monte Carlo error in any numerical simulation. To apply this formula it should be checked that the evolution equation \eqref{e:evol_pot} and the no-leakage principle are complete, i.e.\ contain all relevant flows, else these should be amended to arrive at an augmented representation. Numerical evidence for the reduced Monte Carlo error  is provided in Table~\ref{tab:num_res} where it can be seen that 
$FDB_{rep}^{mce}$, corresponding to representation~\eqref{e:fdb-rep}, is significantly smaller than $FDB_{CF}^{mce}$. 
\item 
The representation formula also gives conditions on the suitable market consistency  of an interest rate scenario generator: the expressions $E[B_t^{-1}cog_t]$ should be priced as accurately as possible.  To connect this to market data we can use the simplified model \eqref{e:gshat}, acting as an over-estimate (cf.~\eqref{e:ass8test}), to arrive at the caplet prices $\mathcal{O}_t^-$ defined in \eqref{e:O}. The market implied value, $\widehat{COG}$, can now be compared to the model implied value, $\widehat{COG}_{MC}$. These values are contained in Table~\ref{tab:num_res}, and hence it can be observed that the MF-LMM provides an acceptable level of market consistency (compared to $MV_0$) but certainly not an exact fit.   
\item 
Section~\ref{sec:numALM} contains a detailed description of management rules. These rules are designed to maximize shareholder value while meeting certain constraints (most prominently the surplus fund constraint \eqref{e:theta-con}) to realistically remain a competitive insurance provider. However, this maximization is only speculative and not proven in any formal sense. In this regard it would be interesting to analyze the connection to the optimal control problem stated in Section~\ref{sec:con}, this is left as an open problem for future investigation.
\item 
Section~\ref{sec:assump} contains the assumptions and numerical study providing evidence for their validity so that we may derive lower and upper bound estimates in Section~\ref{sec:est}. The relevant formulae are \eqref{e:LB}, \eqref{e:UB} and \eqref{FDBhat} for $\widehat{LB}$, $\widehat{UB}$ and $\widehat{FDB}$, respectively. The crucial point of these formulae is that their calculation is very easy, it is purely algebraic and depends only on a few numbers. Table~\ref{tab:num_res} shows that the estimates hold over a relatively wide range of parameters. We have  considered $9$ scenarios corresponding to different levels of initial unrealized gains and scalings of premium payments. The latter is equivalent to varying the relationship between prevailing yield curve and level of minimum guarantee rate. 
\item 
Section~\ref{sec:All_est} contains an application of the estimation formulae to public data of a real life insurance company for \revise{$7$} different accounting years. These results are summarized in Table~\ref{tab:0}, and it can be observed that the estimation is successful to a remarkable degree of accuracy -- given that the \revise{$7$} accounting years represent quite different economic conditions and that all the relevant information is taken from publicly available records. These sources are provided en d\'etail in Tables~\ref{table:values} and \ref{tbl:source}, and have to be complimented by the prevailing yield curve (\cite{RFR}). 
\end{enumerate}


\section*{Declarations}
\subsection*{Funding} 
This work is not supported by external funding. 
\subsection*{Conflicts of interest/Competing interests} 
The authors declare no conflicts of interest or competing interests.
\subsection*{Availability of data and material} 
All data used in this work is either publicly available (sources are given in the References) or aggregated and anonymized. Data can be provided upon request to the corresponding author.  
\subsection*{Authors' contributions} 
Florian Gach analyzed most of the public data. Simon Hochgerner designed the project and wrote most of the manuscript. Eva Kienbacher and Gabriel Schachinger carried out most of the programming. All authors reviewed the manuscript.



\end{document}